\documentclass[a4paper, UKenglish, cleveref, autoref, thm-restate]{lipics-v2021}

\pdfoutput=1 
\hideLIPIcs  

\title{Multipass Linear Sketches for Geometric LP-Type Problems}

\author{N. Efe Çekirge}{Department of Computer Science, Dartmouth College, Hanover, NH, USA\footnote{This work was done while at Carnegie Mellon University.}}{nabi.efe.cekirge.gr@dartmouth.edu}{https://orcid.org/0009-0007-6531-7593}{}

\author{William Gay}{Grainger College of Engineering, University of Illinois, Urbana-Champaign, IL, USA\footnote{This work was done while at Carnegie Mellon University.}}{whgay2@illinois.edu}{}{}

\author{David P. Woodruff}{Computer Science Department, Carnegie Mellon University, Pittsburgh, PA, USA}{dwoodruf@andrew.cmu.edu}{https://orcid.org/0000-0002-2158-1380}{Supported by a Simons Investigator Award and Office of Naval Research award number N000142112647.}

\authorrunning{N.\,E. Çekirge and W. Gay and D.\,P. Woodruff}

\Copyright{N. Efe Çekirge and William Gay and David P. Woodruff}



\nolinenumbers

\EventEditors{}
\EventNoEds{0}
\EventLongTitle{}
\EventShortTitle{}
\EventAcronym{}
\EventYear{}
\EventDate{}
\EventLocation{}
\EventLogo{}
\SeriesVolume{}
\ArticleNo{}

\usepackage{tcolorbox}
\usepackage[ruled, linesnumbered]{algorithm2e}

\let\Pr\relax
\DeclareMathOperator*{\Pr}{\mathbf{Pr}}

\begin{document}


\newcommand{\bmU}{\boldsymbol{U}}

\newcommand{\bma}{\boldsymbol{a}}
\newcommand{\bmc}{\boldsymbol{c}}
\newcommand{\bme}{\boldsymbol{e}}
\newcommand{\bmp}{\boldsymbol{p}}
\newcommand{\bmq}{\boldsymbol{q}}
\newcommand{\bmu}{\boldsymbol{u}}
\newcommand{\bmv}{\boldsymbol{v}}
\newcommand{\bmx}{\boldsymbol{x}}
\newcommand{\bmy}{\boldsymbol{y}}
\newcommand{\bmz}{\boldsymbol{z}}

\newcommand{\mcA}{\mathcal{A}}
\newcommand{\mcB}{\mathcal{B}}
\newcommand{\mcC}{\mathcal{C}}
\newcommand{\mcI}{\mathcal{I}}
\newcommand{\mcJ}{\mathcal{J}}
\newcommand{\mcP}{\mathcal{P}}
\newcommand{\mcQ}{\mathcal{Q}}
\newcommand{\mcR}{\mathcal{R}}

\newcommand{\bbN}{\mathbb{N}}
\newcommand{\bbR}{\mathbb{R}}

\newcommand{\eps}{\varepsilon}

\newcommand{\epsnet}{$\eps$-net}
\newcommand{\epsilonnet}{$\epsilon$-net}
\newcommand{\munet}{$\mu$-net}

\newcommand{\Ai}{A^{(i)}}
\newcommand{\bi}{b^{(i)}}
\newcommand{\ei}{e^{(i)}}
\newcommand{\bmpi}{\bmp^{(i)}}
\newcommand{\bmpj}{\bmp^{(j)}}
\newcommand{\bmqj}{\bmq^{(j)}}

\newcommand{\minds}{\min(d, S)}
\newcommand{\mindf}{\min\paren{\sqrt{d}, F}}

\newcommand{\abs}[1]{\left|{#1}\right|}
\newcommand{\paren}[1]{\left({#1}\right)}
\newcommand{\brak}[1]{\left[{#1}\right]}
\newcommand{\cbrak}[1]{\left\{{#1}\right\}}
\newcommand{\norm}[1]{\left\lVert{#1}\right\rVert}
\newcommand{\ceil}[1]{\left\lceil{#1}\right\rceil}
\newcommand{\ip}[2]{\left\langle #1, #2 \right\rangle}
\newcommand{\pipe}[1]{\left|{#1}\right|}

\newcommand{\sfrac}[2]{#1/#2}

\newcommand{\Tr}{\operatorname{Tr}}
\newcommand{\Frobinner}[2]{\langle #1 , #2 \rangle_F}
\newcommand{\Vectorize}{\operatorname{Vec}}

\newcommand{\poly}{\operatorname{poly}}
\newcommand{\polyof}[1]{\poly \left( #1 \right)}

\newcommand{\polylog}{\operatorname{polylog}}
\newcommand{\polylogof}[1]{\polylog \left( #1 \right)}

\newcommand{\TV}{\mathrm{TV}}

\newcommand{\PrP}{\Pr\nolimits_{\mcP}}
\newcommand{\PrQ}{\Pr\nolimits_{\mcQ}}
\newcommand{\Exp}[2]{\mathop{\mathbb{E}}_{#1}\left[{#2}\right]}


\maketitle

\begin{abstract}
LP-type problems such as the Minimum Enclosing Ball (MEB), Linear Support Vector Machine (SVM), Linear Programming (LP), and Semidefinite Programming (SDP) are fundamental combinatorial optimization problems, with many important applications in machine learning applications such as classification, bioinformatics, and noisy learning. We study LP-type problems in several streaming and distributed big data models, giving $\varepsilon$-approximation linear sketching algorithms with a focus on the high accuracy regime with low dimensionality $d$, that is, when ${d < (1 / \varepsilon)^{0.999}}$. Our main result is an $O(ds)$ pass algorithm with $O(s( \sqrt{d}/\varepsilon)^{3d/s}) \cdot \mathrm{poly}(d, \log (1/\varepsilon))$ space complexity in words, for any parameter $s \in [1, d \log (1/\varepsilon)]$, to solve $\varepsilon$-approximate LP-type problems of $O(d)$ combinatorial and VC dimension. Notably, by taking $s = d \log (1/\varepsilon)$, we achieve space complexity polynomial in $d$ and polylogarithmic in $1/\varepsilon$, presenting exponential improvements in $1/\varepsilon$ over current algorithms. We complement our results by showing lower bounds of $(1/\varepsilon)^{\Omega(d)}$ for any $1$-pass algorithm solving the $(1 + \varepsilon)$-approximation MEB and linear SVM problems, further motivating our multi-pass approach.
\end{abstract}

\section{Introduction}

LP-type problems are a class of problems fundamental to the field of combinatorial optimization. Introduced originally by Shahir and Welzl, they are defined by a finite set of elements $S$ and a solution function $f$ which maps subsets of $S$ to a totally ordered domain and also satisfies the properties of locality and monotonicity~\cite{SW92}. Examples of LP-type problems include the Minimum Enclosing Ball (MEB) problem, the Linear Support Vector Machine(SVM) problem, and linear programming (LP) problems for which the class is named. These problems have many applications in different fields such as machine learning. For example, the MEB problem has applications in classifiers~\cite{BJS03}, and support vector machines~\cite{B98, TKC05}. The linear SVM problem is a fundamental problem in machine learning, used in many classification problems such as text classification, character recognition, and bioinformatics~\cite{CGLRML20}. Linear programs have many uses in data science problems, with \cite{MVPV18} showing that most common machine learning algorithms can be expressed as LPs. One such example is the linear classification problem, which has applications in many noisy learning and optimization problems~\cite{B94, BFKV96, DV08}. A related class of problems are semidefinite programming (SDP) problems, which are useful for many machine learning tasks such as distance metric learning~\cite{XJRN02} and matrix completion~\cite{CT10,R11,CR12}.

For these applications, the size of the input problem can be very large, and often too expensive to store in local memory. Further, the data set can be distributed between multiple machines which require a joint computation. In these applications, the large size of data generally means that finding exact solutions to problems is prohibitively expensive in space or communication complexity, and such many algorithms instead keep and perform computations on a small representation of the data which allows them to give an answer within a small approximation factor of the original solution.

One such way to represent data is with linear sketches. For a large dataset $P$ of $n$ elements from a universe of size $U$, which can be thought of as a vector $\bmU \in \cbrak{0, 1}^U$ where $\bmU_{\bmp}$ is $1$ for each element $\bmp \in P$ and $0$ otherwise, a linear sketch generally is in the form of an $N \times U$ matrix $A$, for $N \ll n$. Thus, one only has to store $A \cdot \bmU$, a vector of size $N$, which is much cheaper than storing the full $n$ element data set. Linear sketches are very useful in many big data applications, such as data streams where data points are read one by one, or parallel computation where the data set is split among different machines. This is because updating the data vector $\bmU \leftarrow \bmU + \bme_{\bmp}$ for a newly read data point $\bmp$, where $\bme_{\bmp}$ is the $\bmp$\textsuperscript{th} standard unit vector, simply requires updating the sketch as $A \cdot \bmU + A_{\cdot \bmp}$ for the $\bmp$\textsuperscript{th} column of $A$, $A_{\cdot \bmp}$. Similarly, for data distributed among $k$ machines as $P = \sum_{i \in \brak{k}} P_k$, the sketch can be maintained separately by each machine, since $A \cdot P = \sum_{i \in \brak{k}} A \cdot P_i$. As linear sketches consist of matrix-vector products, they are useful in GPU-based applications, which are particularly suited for these operations. Linear sketching algorithms have extensive study in big data applications, see~\cite{W14, LNW14, AHLW16} for some examples and surveys. An extension which we build our algorithm in is the multiple linear sketch model, which corresponds to making adaptive matrix-vector products in each of a small number of rounds~\cite{SWYZ21}.

\subsection{Our Contributions}

In this paper, we explore using linear sketches to solve approximate LP-type problems in high accuracy regimes, in which $d < \paren{\sfrac{1}{\eps}}^{0.999}$, where $d$ is the dimensionality of the problem and $\eps$ is the approximation accuracy. This regime is applicable to many real world big data applications where $\sfrac{1}{\eps}$ might be several orders of magnitude larger than $d$, such as machine learning tasks of regression or classification. Take for example~\cite{NN04}, where $d = 2$, $\sfrac{1}{\eps} \in \brak{10^2, 10^5}$, and $n = 10^5$, or~\cite{TKC05} where for some data sets $\sfrac{1}{\eps}$ is taken to be much larger than $d$ yet small compared to $n$.

We present an algorithm in the multiple linear sketch model for solving certain classes of LP-type problems, which can be utilized in several prominent big data models. The multiple linear sketch model, which we define further in Section~\ref{relatedworksection}, can be thought of as a multipass extension of the linear sketch model, where one is allowed to choose a fresh linear sketch in each pass of a multipass streaming algorithm in order to have a better representation of the data.
There are two significant benefits of our sketching approach. Firstly, our algorithm works in the presence of input point deletions, such as in the turnstile model.
Secondly, it allows our algorithm to work in the presence of duplicated points in the input, which generally can prove problematic for sampling based algorithms.

We are given a universe $S = \cbrak{-\Delta, -\Delta + 1, \hdots, \Delta}^d$ of $d$-dimensional points, where each point can be represented by a word of size $d \log \paren{\Delta n}$. This defines our cost model, and we give our space complexity in the number of points/words. We use a linear sketch to reduce an input set $P \subseteq S$ of $n$ input points into a point set $Q = \cbrak{(1 + \eps)^l \cdot \bme \mid \bme \in E_{\bmp_c}, l \in \brak{\log_{1 + \eps} \delta}}$, where $E_{\bmp_c}$ is a metric \epsnet{} centered at a point $\bmp_c \in P$ and $\delta$ is an $O(1)$ approximation for the distance of the furthest point in $P$ to $\bmp_c$.
Our algorithm is also applicable when $S = \cbrak{-1, -1 + \frac{1}{\Delta}, \hdots, 1}^d$. Now, we reduce $P$ directly into the point set $Q = E$ where $E$ is the metric \epsnet{} centered at the origin.
Intuitively, this has the effect of snapping each input point to the closest of a small number of net points, thereby reducing our input space.
This linear sketch can be though of as an $N \times \abs{S}$ matrix $A$ for $N = l \abs E \ll n$ where each row corresponds to one possible point $(1 + \eps)^l \cdot \bme$, containing 1's for all the points $\bmp \in S$ with direction closest to $\bme$ and norm closest to $(1 + \eps)^l$.
The major benefit of using this linear sketch is that we don't actually store our metric \epsnet{}. We can calculate where each point is snapped to quickly on the fly, so we do not suffer the exponential in $d$ space complexity of $\eps$-kernel based formulations.

This linear sketch of the data can now be used to run our algorithm on, treating it as a smaller virtual input compared to the original input. Our general algorithm for solving LP-type problems follows the sampling idea presented by~\cite{AKZ19}, which in turn is based on Clarkson's algorithm for linear programs in low dimensions~\cite{C95}, where input points are assigned weights, and a weighted sampling procedure is used in order to solve the LP-type problem on smaller samples from the input. By modifying the weights throughout multiple iterations, the algorithm quickly converges on samples that contain a basis for the LP-type problem, for which the solution is the solution of the entire input set.

As our algorithm is based on linear sketches, we can utilize it in many big data models.
We can solve problems on input streams that we can make multiple linear scans over in the \emph{Multipass Streaming} model. Additionally, our algorithm works for streams with deletions, such as in the \emph{Multipass Strict Turnstile} model. Finally, we can also formulate our algorithm for distributed inputs, minimizing communication load in the \emph{Coordinator} and \emph{Parallel Computation} models. A usage of these models is seeing the symmetric difference of data over time across multiple servers. Over two servers with large datasets, communication can be thought of as passes over a stream with deletions. Further, we can find the symmetric difference of server datasets in the Parallel Computation model as well. This is helpful in capturing change in data to predict current trends, see e.g.~\cite{CM03}.

We give the following result for LP-type problems that are bounded in combinatorial and VC-dimension. More detailed bounds and time bounds are given in Sections~\ref{multipasssection},~\ref{turnstilesection},~\ref{coordinatorsection}, and~\ref{parallelsection}.
\begingroup \abovedisplayskip=2pt \belowdisplayskip=0pt
\begin{tcolorbox}
[width=\linewidth, sharp corners=all, colback=white, boxrule=0.4pt]
\begin{theorem}\label{mainresult}
    For any $s \in [1, \log N]$, there exists a randomized algorithm to compute a ${(1 + O(\eps))}$-approximation solution to an LP-type problem with VC and combinatorial dimensions in $O(d)$ with high probability, using a linear sketch of dimensionality $N$ where $\log N \in \polyof{d, \log \paren{\sfrac{1}{\eps}}}$, where $d\log(n \Delta)$ is the bit complexity to store one word/point $\bmp \in P$, and where a solution $f(\cdot)$ can be stored in $O(1)$ words. This algorithm can be implemented in various models as such.
    \begin{itemize}
        \item \textbf{Multipass Streaming}: An $O\paren{d s}$ pass algorithm with $O \paren{s N^{3/s}}\cdot \polyof{d, \log N}$
        space complexity in words.
        
        \item \textbf{Strict Turnstile}: An $O\paren{d s + \log \log \paren{\Delta d}}$ pass algorithm with $O \paren{s N^{3/s}}\cdot \polyof{d, \log N}$
        space complexity in words.

        \item \textbf{Coordinator}: An $O \paren{ds}$ round algorithm with $O \paren{\polyof{d, N^{\sfrac{1}{s}}} + k}$
        load in words.

        \item  \textbf{Parallel Computation}: An $O \paren{ds}$ round algorithm with $O \paren{\polyof{d, N^{\sfrac{1}{s}}} + k}$
        load in words.
    \end{itemize}
\end{theorem}
\end{tcolorbox}
\endgroup
For most problems, we have that $N \in \paren{\sfrac{1}{\eps}}^{O(d)}$. With this, we get the following corollary.
\begin{corollary}
    For any $s \in \brak{1, d \log \paren{\sfrac{1}{\eps}}}$, there exists a randomized algorithm to compute a ${(1 + O(\eps))}$-approximation solution to an LP-type problem with VC and combinatorial dimensions in $O(d)$ with high probability, taking $O\paren{d s}$ passes and
    $O \bigl( s \bigl(\sqrt{d}/\eps\bigr)^{3d/s}\bigr) \cdot \polyof{d, \log \paren{\sfrac{1}{\eps}}}$
    space complexity in words.
\end{corollary}
Our algorithm provides powerful results when we set the parameter $s$ equal to $\log N$. As we have that $\log N \in \polyof{d, \log \paren{\sfrac{1}{\eps}}}$, this allows us to achieve algorithms with pass and space complexities polynomial in $d$ and polylogarithmic in $\sfrac{1}{\eps}$.

We can use this general algorithm to find approximate solutions for many LP-type problems which satisfy some elementary properties, as described in Section~\ref{algosection}. We emphasize that most LP-type problems fulfill these properties. We present applications of our general algorithm for the following problems. Note that for the Linear SVM, Bounded LP, and Bounded SDP problems, we do not have the $\log \log (\Delta d)$ pass increase in the strict turnstile model.
\begin{itemize}
    \item \textbf{MEB}: Given $n$ points $P \subseteq \cbrak{-\Delta, -\Delta + 1, \hdots, \Delta}^d$, our objective is to find a minimal ball enclosing the input points. We present an algorithm to output a ball that encloses $P$ with radius at most $(1 + O(\eps))$ times optimal. For the MEB problem, we have $N \in \paren{\sfrac{1}{\eps}}^{O(d)}$.

    \item \textbf{Linear SVM}: Given $n$ labeled points $P \subseteq \cbrak{-1. -1 + \frac{1}{\Delta}, \hdots, 1}^d \times \cbrak{\pm1}$ and some $\gamma > 0$ for which the points are either linearly inseparable or $\gamma$-separable, our objective is to find a separating hyperplane of $P$ with the largest margin. We present an algorithm to output either that the points are linearly inseparable, or a separating hyperplane $(\bmu, b)$ where $\norm{\bmu}^2$ is at most $(1 + O(\eps))$ times optimal. For the linear SVM problem, we have $N \in \paren{\sfrac{1}{\eps \gamma}}^{O(d)}$.

    \item \textbf{Bounded LP}: Given an objective vector $\bmc \in \cbrak{-1, -1 + \frac{1}{\Delta}, \hdots, 1}^d$ such that $\norm{\bmc} \in O(1)$ and $n$ constraints $P \subseteq \cbrak{-1. -1 + \frac{1}{\Delta}, \hdots, 1}^{d + 1}$, where for each constraint $(\bma_i, b_i) \in P$, $\norm{\bma_i}, \abs{b_i} \in O(1)$, our objective is to find a solution $\bmx$ with $\norm{\bmx} \in O(1)$ that maximizes $\bmc^T \bmx$ while satisfying the constraints. We present an algorithm to output a solution $O(\eps)$ within each constraint, and for which the objective is at most $O(\eps)$ smaller than optimal. For bounded LP problems, we have $N \in \paren{\sfrac{1}{\eps}}^{O(d)}$.

    \item \textbf{Bounded SDP}: Given an objective matrix $C \in \cbrak{-1. -1 + \frac{1}{\Delta}, \hdots, 1}^{d \times d}$ such that ${\norm{C}_F \leq 1}$ and $n$ constraints $P \subseteq \cbrak{-1. -1 + \frac{1}{\Delta}, \hdots, 1}^{d \times d} \times \cbrak{-1, -1 + \frac{1}{\Delta}, \hdots, 1}$, where for each constraint $(\Ai, \bi) \in P$, $\norm{\Ai}_2, \abs{\bi} \leq 1$, $\norm{\Ai}_F \leq F$, and the number of nonzero entries of $\Ai$ is bounded by $S$, our objective is to find a positive semidefinite solution $X$ of unit trace that maximizes the objective $\Frobinner{C}{X}$ while satisfying the constraints. We present an algorithm to output a solution that is $O(\eps)$ within each constraint, and for which the objective is at most $O(\eps)$ smaller than optimal. As we are working with $d \times d$ matrices, the combinatorial and VC dimensions are $O(d^2)$, and we have $N \in \binom{d^2}{S} \cdot \paren{\sfrac{1}{\eps}}^{O(S)} + \paren{\sfrac{1}{\eps}}^{O(d)}$. $\norm{\cdot}_2$, $\norm{\cdot}_F$, and $\Frobinner{\cdot}{\cdot}$ refer to the spectral and Frobenius norms and Frobenius inner product respectively, explained further in Appendix~\ref{sdpsubsection}.
\end{itemize}
Table~\ref{resultstable} compares our results to those of prior work, including the \emph{exact} LP-type problem algorithm of~\cite{AKZ19}, and various $\eps$-approximation LP-type problems.

\begin{table}
    \caption{Result Comparisons for Various LP-type Problems} 
    \label{resultstable}
    \centering

    \begin{tabularx}{\textwidth}{llXX}
    \hline
    \textbf{Problem} & \textbf{Ref.} & \textbf{Pass Complexity} & \textbf{Space Complexity (words)}\\
    \hline
    \multirow{2}{*}{LP-Type Problems} & \cite{AKZ19}& $O(d \log n)$ & $\polyof{d, \log n}$\\
    & Ours & $O \paren{d^2 \log \paren{\sfrac{1}{\eps}}}$ & $\polyof{d, \log \paren{\sfrac{1}{\eps}}}$\\
    \rule{0pt}{6ex}
    \multirow{2}{*}{LP} & \cite{LSZZZ23} & $\Tilde{O}\paren{\sqrt{d}\log \paren{\sfrac{1}{\eps}}}$ & $\Tilde{O} \paren{d^2}$\\
    & Ours & $O \paren{d^2 \log \paren{\sfrac{1}{\eps}}}$ & $\polyof{d, \log \paren{\sfrac{1}{\eps}}}$\\
    \rule{0pt}{6ex}
    \multirow{3}{*}{MEB} & \cite{Z11} & 1 & $O \paren{\paren{\sfrac{1}{\eps}}^{\frac{d - 1}{2}} \log \paren{\sfrac{1}{\eps}}}$\\
    & \cite{BC03} & $\ceil{\sfrac{2}{\eps}}$ & $O \paren{\sfrac{d}{\eps}}$\\
    & Ours & $O \paren{d^2 \log \paren{\sfrac{1}{\eps}}}$ & $\polyof{d, \log \paren{\sfrac{1}{\eps}}}$\\
    \rule{0pt}{6ex} 
    \multirow{2}{*}{Linear Classification} & \cite{CHW12}& $\Tilde{O}\paren{\sfrac{1}{\eps}}$ & $\Tilde{O}\paren{\sfrac{1}{\eps^2}}$\\
    & Ours & $O \paren{d^2 \log \paren{\sfrac{1}{\eps}}}$ & $\polyof{d, \log \paren{\sfrac{1}{\eps}}}$\\
    \rule{0pt}{6ex} 
    \multirow{4}{*}{Bounded SDP} & \cite{SYZ23} & $O \paren{\sqrt{d} \log \paren{\sfrac{1}{\eps}}}$ & $O\paren{n^2 + d^2}$\\
    & Ours & $O \paren{d^4 \log \paren{\sfrac{1}{\eps}}}$ & $\polyof{d, \log \paren{\sfrac{1}{\eps}}}$\\
    & \cite{GH16} & \multicolumn{2}{l}{$O \paren{\frac{1}{\eps^2} \log n \paren{F^2 (n + \log d) + \frac{1}{\eps} S \log d \atop + \min \paren{\frac{1}{\eps^2} S \log d, d^2} \frac{1}{\eps^2} \log d}}$ \textbf{time complexity}}\\
    & Ours & \multicolumn{2}{l}{$\Tilde{O}\paren{d^2 S \log \paren{\frac{1}{\eps}} \paren{n + d^{4.7437}}}$ \textbf{time complexity}}\\
    \hline
    \end{tabularx}
\end{table}

For the approximate MEB problem, there are many results, as discussed in Section~\ref{relatedworksection}.
One such is the Zarrabi-Zadeh result, which is among many that utilize an $\eps$-kernel~\cite{Z11}. While these algorithms are single pass, they require space complexity of $\paren{\sfrac{1}{\eps}}^{O(d)}$. Thus, we present a large increase in space efficiency over this method. We also present a lower bound on the space complexity for any one-pass algorithm for the MEB problem of $\paren{\sfrac{1}{\eps}}^{\Omega(d)}$ in Section~\ref{lowerboundsection}, which further motivates our multi-pass approach. While the result by B\^{a}doiu and Clarkson has space complexity polynomial in $d$, it is also polynomial in $\sfrac{1}{\eps}$ in passes and space~\cite{BC03}. Thus, in the high accuracy regime, we present an exponential improvement over their result.

Another approach to the $(1 + \eps)$-MEB problem is by using the ellipsoid algorithm with a similar metric \epsnet{} construction. While this algorithm also can achieve $\polyof{d, \log \paren{\sfrac{1}{\eps}}}$ pass and space complexity, we present two main improvements over using this approach. Most importantly, as our algorithm is in the multiple linear sketch model, we are able to apply our result to streams with addition and removal of points, and the coordinator and parallel computation models, where the input is distributed. As the ellipsoid algorithm is not in the multiple linear sketch model, it cannot, for example, find the MEB of the symmetric difference of pointsets from two machines. Further, the bulk of our algorithm utilizes $\ell_0$ samplers and estimators and simple arithmetic, which can allow our algorithm to be practical to implement.

A strong result to linear programs is by~\cite{LSZZZ23}. They use the interior point method for solving linear programs in order to get their result. There are two significant differences of our result to theirs. Firstly, as with the ellipsoid method, while this result might present an improvement in the multipass streaming model, our algorithm is able to handle point additions and deletions, as well as distributed inputs. Secondly, their tilde notation hides logarithmic dependence on $n$. Our result completely removes this dependence in both the pass and space complexities, showing that $n$ is not an inherent parameter in this problem.

Our improvements over the work by~\cite{AKZ19} is similar to this, where by using linear sketches we are able to extend their results to models with point additions and deletions such as the turnstile model, and additionally deal with duplicate points due to our usage of $\ell_0$ samplers and estimators.

As we present in Section~\ref{linearclassificationsection}, our algorithm can be used to solve the linear classification problem. Note that our result finds an exact classifier. We compare our result to the result by~\cite{CHW12}.
In the high accuracy regime, this represents an exponential improvement over their result in pass and space complexity.

Another application of our framework is for additive $\eps$-approximation bounded SDP problems. As we present in Section~\ref{saddlepointsection}, our formulation can be used to solve for the saddle point problem within the unit simplex. For this problem, there is a result by~\cite{GH16} in the number of arithmetic operations. We achieve an exponential improvement in $\sfrac{1}{\eps}$ dependence over their result, significant in the high accuracy regime.
Another result is by~\cite{SYZ23} who achieve a significant result for SDPs with high dimensionality.
In the high accuracy regime with a large number of constraints $n \gg \sfrac{1}{\eps}$, this presents a loss in number of passes polynomial in $d$, while achieving a significant gain in space complexity, as our result does not depend on $n$ at all.

\subsection{Related Work} \label{relatedworksection}

Approximation algorithms for LP-type problems in big data models have been well-studied, where the trade-off between approximation ratio, number of passes, and space complexity has given rise to different techniques being used for different regimes.

For the MEB problem, there is the folklore result of a $1$-pass algorithm, where one stores the first point and a running furthest point from it, giving a $2$-approximation. There is also a classic $1$-pass $\sfrac{3}{2}$-approximation algorithm by Zarrabi-Zadeh and Chan, that uses $O(d)$ space complexity~\cite{ZC06}. Allowing $O(d)$ passes, an algorithm by Agarwal and Sharathkumar, with further analysis by Chan and Pathak, achieves a $1.22$-approximation with $O \paren{\paren{\sfrac{d}{\eps^3} \log \paren{\sfrac{1}{\eps}}}}$ space complexity~\cite{AS10, CP11}. For the $(1 + \eps)$-approximation problem, much work is done in the single-pass model. These algorithms are based on storing extremal points, called an $\eps$-kernel, and achieve a space complexity of $\paren{\sfrac{1}{\eps}}^{O(d)}$~\cite{AHV04, C06, AY07, Z11}. There is also a body of work on bounding the size of these $\eps$-kernels, as they are useful to achieve low-space representations of large data~\cite{KMY04, BC08, Y08, C10}. \cite{AdFM17} presents a lower bound of $\paren{\sfrac{1}{\eps}}^{\Omega(d)}$ on the size of an $\eps$-kernel. In the multi-pass model, there is the classic result by B\^{a}doiu and Clarkson ~\cite{BC03}. This algorithm is widely used in big data models, for its linear dependence on $\sfrac{d}{\eps}$ for space and $\sfrac{1}{\eps}$ for pass complexity, and simplicity~\cite{NN04, TKC05}.

Linear programs are frequently studied in big data models, with a major focus on exact multi-pass LP results. 
A significant recent result is by~\cite{AKZ19}, who in the multi-pass model achieve a result with $O(d \log n)$ passes and $\polyof{d, \log n}$ space usage. Their result is also applicable to other LP-type problems as well as parallel computation models. 
Their algorithm builds on Clarkson's sampling based method~\cite{C95}, and utilizes reservoir sampling to sample a $\mu$-net~\cite{C82}. We build on the ideas presented by~\cite{AKZ19}, to solve LP-type problem approximations.
There is also recent work on multipass approximations for LPs, with a result by~\cite{LSZZZ23} who achieve an $\Tilde{O} \paren{\sqrt{d} \log \paren{\sfrac{1}{\eps}}}$ pass $\Tilde{O} \paren{d^2}$ space algorithm for solving $(1-\eps)$-approximation LPs.
For packing LPs \cite{AG11} gives a $(1 + \eps)$-approximation algorithm, and for covering LPs ~\cite{IMRUVY17} gives a $(1 + \eps)$-approximation algorithm.
The linear classification problem is also well studied, as a fundamental machine learning problem. Its classic result is the mistake bound of the perceptron algorithm~\cite{N62}. More recently, a result by~\cite{CHW12} achieves $O\paren{\sfrac{1}{\eps^2}\log n}$ iterations to find a linear classifier.

There is a lot of study on approximate semidefinite programs, focusing on different regimes. \cite{GH16} give a result on solving bounded SDP saddle point problems with certain assumptions about the problem input. We use similar assumptions to compare our results to theirs. For the regime with a low number of constraints and a large dimension, \cite{SYZ23} achieves an $O\paren{\sqrt{d} \log{\sfrac{d}{\eps}}}$ pass algorithm with $\Tilde{O}\paren{n^2 + d^2}$ space complexity.

\section{Preliminaries}

\paragraph*{Multiple Linear Sketch Model}
In the multiple linear sketch model, an algorithm is able to create fresh linear sketches of the data through each iteration of a multiple iteration model, such as the multipass streaming or multipass strict turnstile models, based on information from the previous iterations. This model has been used implicitly for many adaptive data stream algorithms, such as finding heavy hitters, sparse recovery, and compressed sensing~\cite{IPW11, PW12, NSWZ18}.

\paragraph*{$\ell_0$ Estimator}
\begin{theorem}[Theorem 10 from~\cite{KNW10}]
    There is a $1$-pass linear sketch which given additive and subtractive updates to an underlying vector $\bmv$ in $\{0, 1, \hdots, \polyof{n}\}^N$, gives a $(1 \pm \zeta)$-approximation to the number of non-zero coordinates of $v$, denoted $\ell_0(v)$, with error probability $\delta$, and using $O(\frac{\log N \log \delta^{-1}}{\zeta^2})$ words of memory.
\end{theorem}

\paragraph*{$\ell_0$ Sampler}
\begin{theorem}[Theorem 2 from~\cite{JST11}]
    There is a $1$-pass linear sketch which given additive and subtractive updates to an underlying vector $\bmv$ in $\{0, 1, \hdots, \polyof{n}\}^N$, outputs a coordinate $i$ in the set of non-zero coordinates of $\bmv$, where for each non-zero coordinate $j$ in $\bmv$, we have $i = j$ with probability $\frac{1}{\ell_0(v)} \pm \delta$, using $O(\log N \log \delta^{-1})$ words of memory.
\end{theorem}

\begin{note}
    Notice that for the $\ell_0$ estimator and sampler, the space complexity in words is logarithmic in the dimensionality of the underlying vector which in our case is $N \ll n$. While the entries of the vector can be polynomial in $n$, the number of words required to store the estimator and sampler are independent of $n$.
\end{note}

\paragraph*{Combinatorial Dimension} The combinatorial dimension of a problem is the maximum cardinality of a basis for the problem, denoted $\nu$. For LP-type problems, we consider a basis as the minimum number of points required to define a solution to the problem.

\paragraph*{VC Dimension} Given a set $X$ and a collection of subsets of $X$ $\mcR$, known as a \emph{set system} $(X, \mcR)$ and $Y \subseteq X$, the projection of $\mcR$ on $Y$ is defined as ${\mcR|_Y := \cbrak{Y \cap R \mid R \in \mcR}}$. The VC dimension of $(X, \mcR)$ is the minimum integer $\lambda$ where $\left|\mcR|_Y\right| < 2^{|Y|}$ for any finite subset $Y \subseteq X$ such that $|Y| > \lambda$~\cite{VC71}.

\paragraph*{Metric \epsnet{}} Given a metric space $(M, d)$, a metric \epsnet{} is a set of points $E \subseteq M$ such that for any point $\bmp \in M$, there exists a point $\bme \in E$ for which $d \paren{\bmp, \bme} \leq \eps$.
For large metric spaces, i.e. when $M = \cbrak{-\Delta, -\Delta + 1, \hdots, \Delta}^d$, this formulation requires too many points. In this case, we give an alternate formulation. In the case where $M$ is large, given a center point $\bmp_c$, which we consider to be the origin, and a corresponding norm $\norm{\bmq} = d(\bmq, \bmp_c)$ for each point $\bmq \in M$, a metric \epsnet{} is a set of points where for any point $\bmq \in M$, there exists a point $\bme \in E$ for which $d(\bmq, \norm{\bmq}\bme) \leq \eps \norm{\bmq}$.

\paragraph*{\epsilonnet{}} Given a set system $(X, \mcR)$, a weight function $w: X \rightarrow \bbR$, and a parameter $\epsilon \in [0, 1]$, a set $N \subseteq X$ is an \epsilonnet{} with respect to $w(\cdot)$ if $N \cap R \neq \emptyset$ for all sets $R \in \mcR$ such that $w(R) \geq \epsilon \cdot w(X)$~\cite{HW86}.

\begin{lemma}[Corollary 3.8 from \cite{HW86}]\label{mdrawslemma}
    For any set system $(X, \mcR)$ of VC dimension $d < \infty$, finite $A \subseteq X$, and $\epsilon, \delta > 0$, if $N$ is the set of distinct elements of $A$ obtained by \[
        m \geq \max \paren{\frac{4}{\epsilon} \log \frac{2}{\delta}, \frac{8d}{\epsilon} \log \frac{8d}{\epsilon}}
    \]
    random independent draws from $A$, then $N$ is an \epsilonnet{} for $\mcR$ with probability $1 - \delta$.
\end{lemma}

While this construction in its original presentation by Haussler and Welzl~\cite{HW86} applies to the trivial weight function $w(x) = 1$, by drawing with probability proportional to an element's  weight, an \epsilonnet{} with respect to $w(\cdot)$ can be obtained for any weight function~\cite{AKZ19}.

\begin{note*}
    Metric $\eps$-nets are a special case of $\epsilon$-nets, but we prefer to distinguish them for clarity. Further, to not cause confusion between metric $\eps$-nets and $\epsilon$-nets, we refer to the latter as $\mu$-nets.
\end{note*}

\section{Algorithm} \label{algosection}
In this section, we present our algorithm for Theorem~\ref{mainresult}, in the multiple linear sketch model. It is inspired by the algorithm presented in~\cite{AKZ19}, which in turn builds on the ideas of Clarkson in~\cite{C95}. For a constrained optimization problem such as Linear Programming of small dimensionality, Clarkson's presented idea is to repeatedly randomly sample a subset of the constraints using a weighted sampling method. This allows the algorithm to solve a smaller instance of the problem instead, which allows for smaller space usage. The key idea is that each iteration of the algorithm updates the weights multiplicatively which allows convergence into better samples, ending with a sample containing a basis. \cite{AKZ19} utilizes this procedure, with the additional usage of a \munet{} in their sampling procedure allowing for an extension to general LP-type problems of small dimensionality. They also build their weight function by storing solutions from previous passes, thus keeping their storage need small.

These previous algorithms are exact, and thus require pass and space complexities dependent on the number of points in the original stream, $n$. We remove this dependence on $n$ by employing the use of linear sketches, where we snap the space of input points to a smaller space of discrete points formed by a metric \epsnet{}. This allows us to compute $\eps$-approximation solutions to the problems while using much lower storage complexity.

As we utilize the ideas of~\cite{AKZ19} to solve LP-type problems, we retain the needed properties for LP-type problems from their result, presented here as properties~\ref{P1} and~\ref{P2}. Further, we add a third property that must be fulfilled, property~\ref{P3}. We stress that most natural LP-type problems follow properties~\ref{P1} and~\ref{P2}, and that property~\ref{P3} follows naturally for most if not all $\eps$-approximation solutions to these problems. Our algorithm requires an LP-type problem $(S, f)$ to satisfy the following:
\begin{enumerate}[(P1)]
    \item\label{P1} Each element $\bmx \in S$ is associated with a set of elements $S_{\bmx} \subseteq R$ where $R$ is the range of $f$.
    \item\label{P2} For all $P \subseteq S$, $f(P)$ is the minimal element of $\bigcap_{\bmx \in P} S_{\bmx}$.
    \item\label{P3} For all $P, Q \subseteq S$ where $Q$ is the set of points resultant of snapping all the points $\bmp \in P$ to a metric \epsnet{}, a $(1 + O(\eps))$-approximation to $f(Q)$ is a $(1 + O(\eps))$-approximation to $f(P)$.
\end{enumerate}

As an example, consider the MEB problem where $S$ is the set of all $d$-dimensional points, $R$ is the set of all $d$-dimensional balls, and $f$ is the function that returns the MEB of its input. For properties~\ref{P1} and~\ref{P2}, each point is associated with the set of all balls that enclose it. For any set of points, the intersection of these sets of balls are precisely the set of balls that enclose them, and thus their MEB is the minimal element of this set. Property~\ref{P3} follows intuitively because ``unsnapping'' a set of points from a metric \epsnet{} can only increase the size of their MEB by a small amount. A formal proof of this property for the MEB problem is given in Appendix~\ref{MEBapplication}.

Our algorithm is as such. First, we define the metric \epsnet{} we use. We create a lattice covering the unit $d$-cube (which contains the unit $d$-sphere) by allowing, for each $d$ dimensions, the values $-1 + \sfrac{i\eps}{\sqrt{d}}$ where $i$ ranges in $\bigl[0, \sfrac{2\sqrt{d}}{\eps}\bigr]$. Thus, any point that is scaled to the unit $d$-sphere will be at most $\sfrac{\eps}{2}$ away from a point in our metric \epsnet{}. The size of this metric \epsnet{} is $\paren{1 + \frac{2\sqrt{d}}{\eps}}^d$.

The crux of our algorithm is therefore to use this metric \epsnet{} to create a linear sketch of the input. Given any point in the input $\bmp \in P$, we get the point $\frac{\bmp}{\norm{\bmp}}$ and snap it to its closest metric \epsnet{} point. Now, we scale it back up, but we round down the norm to $\paren{1 + \eps}^l$ for some $l \in \bbN$. This has the effect of discretizing our point space along the metric \epsnet{} directions, so that our linear sketch has size bounded by $\ceil{\log_{1 + \eps} \paren{\max_{\bmp \in P} \norm{\bmp}}} \paren{1 + \frac{2 \sqrt{d}}{\eps}}^d$. In some problems, we are presented input points already inside the unit $d$-cube, or a $d$-cube of bounded size. In these cases, we only need to snap the input points to the closest metric \epsnet{} point. We denote the size of our net as $N$ and maintain that it is small compared to the size of the input, $n$.
We stress that we never have to store this net, as the net points are immediately fed into a linear sketch which projects the set of net points to a sketching dimension which is polylogarithmic in the number of net points.

\begin{note*}
    For the sake of clarity, we will refer to the input points $\bmp \in P$ as original points, and these new points $\bmq \in Q$ as snapped points.
\end{note*}

Before we explore the algorithm proper, we note a side effect of creating this sketch. Shrinking the point space  means that distinct points in the original input can now map to the same point. As our algorithm includes weighted sampling over the snapped points, we must treat these as one point and not two overlapping points. In order to achieve this, we utilize $\ell_0$ estimators and samplers, as defined in~\cite{KNW10} and~\cite{JST11} respectively, which allow us to consider estimates on points and sample from them without running into the problem of overlapping. We note that these $\ell_0$ estimators and samplers are linear sketches, and such are able to be used in the models present in our algorithm.

Our algorithm then proceeds in iterations of weighted sampling, in a similar fashion to~\cite{AKZ19}. Throughout, we maintain a weight function $w : Q \rightarrow \bbR$.
In order to save space, we maintain this weight function not explicitly, but by storing previous successful solutions, as in~\cite{AKZ19}.
Each snapped point starts at weight $1$.
In each iteration, we sample a \munet{} with respect to $w(\cdot)$ using Lemma~\ref{mdrawslemma}. As we explain previously, this sampling cannot be done directly from the snapped points $Q$ because of overlaps, so we utilize $\ell_0$ estimators and samplers.

A key point is that each snapped point can increase in weight at most once in each iteration, so we have a number of discrete weight classes.
So, in order to sample a point with accordance to its weight, we use the following procedure.
At each iteration, we create estimators and samplers for each possible weight class. On iteration $t$, this means $t$ copies.
Now, we can estimate the total weight of each weight class, and randomly select a weight class accordingly. Then, since each point in that class have equal weight, we can use our estimator for that weight class to uniformly sample a snapped point. This provides us a weighted sampling procedure.

Given a sample $B$, we calculate its solution $f(B)$, and find the set $V$ of points that violate $f(B)$ (for example, in the MEB problem, violators are points that are outside of the ball). Two sets of $\ell_0$ estimators are now used to calculate estimated weights of the snapped points $Q$ and the violators $V$, overestimating the weight of $Q$ and underestimating the weight of $V$. If the estimated weight of $V$ is at most $\mu$ times the estimated weight of $S$, then we denote this iteration as successful, and multiply the weight of each violator by $N^{\sfrac{1}{s}}$. The intuition behind this is that as our algorithm tries to find a basis for the problem, samples from successful iterations can be thought of as good representations of $Q$, and their violators are more likely to be points in that basis, and should thus be given a higher weight. The algorithm finishes when there are no violators, at which point it has successfully found $f(Q)$. We can then return a $(1 + O(\eps))$-approximation to $f(P)$ using property~\ref{P3}. Further explanation of the algorithm procedure and a proof of correctness is provided in Appendix~\ref{algorithmappendix}.

This sampling procedure is similar to~\cite{AKZ19}, with the major difference being our usage of $\ell_0$ samplers and estimators. This means that we sample according to an approximate weight function rather than the ``true'' weight function. We use a TV distance argument and the fact that we appropriately over/underestimate the weights to argue that we retain the bounds of~\cite{AKZ19}. We show that our approximate weight function is sufficiently close to the ``true'' weight function, enough that a sample of $m$ elements according to the approximate weights will still retain the same properties. Additionally, we account for the error probabilities of the estimators and samplers to show that our probability of being a \munet{} is the same as in the direct sampling of~\cite{AKZ19}. We also show that for actually successful iterations, meaning when the weight of $V$ is at most $\mu$ times the weight of $S$, we maintain this property in our estimated weights. This is because for the second sets of estimators we create, we over/underestimate the weights in order to satisfy this one-sided bound. Thus, we retain the success criteria with high probability, again matching the bounds of~\cite{AKZ19}.

Such, we bound the total number of iterations of our algorithm at $O(\nu s)$ where $\nu$ is the combinatorial dimension of the problem, and $s \in \brak{1, \log N}$ is our parameter. A full explanation is provided in Appendix~\ref{iterationsubsection}.

\subsection{Centering the Net with Deletions}

As our algorithm handles models with point deletions, we must make sure to construct our metric \epsnet{} around a non-deleted point and find an approximation for the distance of the furthest non-deleted point from it, to define the size of our net.
To do this, we use $\ell_0$ samplers, which work in the presence of deletions.
We first sample any non-deleted point, which becomes the center of our net $\bmp_c$.
Now, by construction of $P$, we have an upper and lower bound on the distance of any point from $\bmp_c$.
We simply binary search over the powers of $2$ in this range and sample a non-deleted point with norm larger than the current power of $2$ in each iteration.
By the end of the binary search, which takes $\log \log \paren{\Delta d}$ iterations, we end up with a $2$-approximation for $\delta$, considering only non-deleted points. 
Alternatively, one can find this using $2$ iterations. 
First, similarly we sample a non-deleted input point to become $\bmp_c$. Second, we use an $\ell_0$ sampler for each power of $2$, and in one pass over the points add them to the appropriate sampler. 
This way, instead of doing the steps of a binary search, we perform them all at once. While this negates the extra pass complexity, it adds a space complexity in words that is logarithmic in $\Delta$, since we need to store that many samplers. We will choose the $\log \log \paren{\Delta d}$ iteration solution. We explain this procedure in detail in Section~\ref{turnstilesection}.

\subsection{Sampling from Servers}
When adapting our sampling procedure to distributed models,
one cannot take a sample of $m$ points by asking machines to sample $m$ points each, as this would lead to a load of $m \cdot k$ points. Further, since we are doing weighted sampling, machines need to sample in accordance with other machines' weights. In order to fix these problems, we employ the following protocol, which we describe for the coordinator model, noting that it can be implemented in the parallel computation model by having one machine assume the coordinator role. Each machine sends the coordinator the total weight of their subset. The coordinator then calculates how many of the $m$ sample points $B$ each machine shall sample, and sends them this number. Now, the machines can sample and send only their part of the sample, meaning $m$ points are sent \emph{in total} now, instead of per machine.

\subsection{Approximation Guarantees for Various Problems}
In these big data models, we are able to solve many useful low dimensional LP-type problems. We show how to apply our algorithm in three separate ways for the problems. For the MEB problem, we give a multiplicative $(1 + O(\eps))$-approximation to the optimal ball. For the Linear SVM problem, we show how our algorithm can be presented as a Monte Carlo randomized algorithm. We achieve a one-sided error where if a point set is inseparable, we always output as such, and if a point set is separable, we output a $(1 + O(\eps))$-approximation for the optimal separating hyperplane with high probability. For Linear Programming problems, instead of a multiplicative approximation, we give an additive $\eps$-approximation to bounded LP problems. Full complexity results for the models are presented in their respective subsections of Section~\ref{applicationsappendix}. and an overview is shown in Theorem~\ref{mainresult}.

We can solve SDP problems up to additive $\eps$-approximation. For SDP problems, notice that by taking the vectorization $\bmx$ of a solution $X$, we can rewrite a given SDP problem as an LP in $d^2$ dimensions. We must also define the positive semidefinite constraint $X \succeq 0$. This is equivalent to constraining $\bmy^T X \bmy \geq 0$ for all vectors $\bmy$ where $\norm{y} = 1$. Notice that this requires an infinite number of constraints. To solve this, we create another lattice along with our original metric \epsnet{}. Using this lattice, we can instead use the constraints $\bmz^T X \bmz \geq 0$ for vectors $\bmz$ on the lattice. Just like property~\ref{P3}, we show that satisfying these constraints satisfies the original semidefinite constraints up to $\eps$. Thus, we are able to successfully reduce the SDP problem into an LP, which we can solve up to additive $\eps$-approximation, while retaining a small number of constraints.

Being able to solve bounded LP problems up to additive $\eps$-approximation allows us to use our algorithm for many important related problems. For example, we are able to solve the linear classification problem by finding a classifier within additive $\eps$ of an optimal classifier for a set of points. Since in the problem, the optimal classifier has a separation of at least $\eps$, this means that our found classifier is guaranteed to be a separator for the original points.

\section{Lower Bound} \label{lowerboundsection}
In this section, we motivate the usage of multipass algorithms for achieving subexponential space complexity in $(1 + \eps)$-approximations for LP-type problems in the high accuracy regime, where $d < \paren{\sfrac{1}{\eps}}^{0.999}$. We establish these lower bounds for the MEB and linear SVM problems by analyzing the communication complexity with reductions from the Indexing problem.

It is a well-known result that a lower bound on the $1$-round communication complexity for a problem gives a lower bound on the space complexity of $1$-pass streaming algorithms, via a reduction of Alice running the streaming algorithm on her points and sending the state of the algorithm to Bob, who finishes running the algorithm on his points~\cite{AMS99}. We thus use the standard two party communication model for our lower bounds. A problem in this model is a function $P : \mcA \times \mcB \to \mcC$. Alice receives an input $A \in \mcA$ and Bob receives an input $B \in \mcB$. In an $r$-round protocol, Alice and Bob communicate up to $r$ messages with each other, alternating sender and receiver. In particular, if $r$ is even then Bob sends the first message. If $r$ is odd then Alice sends the first message. After $r$ rounds of communication, Bob outputs some $C \in \mcC$. The goal is for Bob to output $P(A, B)$. The $r$-round communication complexity of a problem $P$ is the minimum worst-case cost over all protocols in which Bob correctly outputs $P(A, B)$ with probability at least $2/3$. In this model, cost is measured by the total number of bits sent between Alice and Bob throughout the $r$ messages.

We use the Indexing problem for our reductions. In it, Alice is given a binary string $b \in \{0, 1\}^{[n]}$ and Bob is given an index $i \in [n]$.
The goal is for Bob to output the $i$\textsuperscript{th} bit of $b$. 
It is well-known that the $1$-round randomized communication complexity of the Indexing problem is $\Omega(n)$.

For the MEB problem, given an instance of the Indexing problem, Alice transforms her bitstring $b$ into a set of points $P$, one point $\bmp_j \in \bbR^d$ for each $b_j = 1$.
Separately, Bob transforms his index $i$ into a point $\bmq \in \bbR^d$.
We then show that the MEB of $P \cup \cbrak{\bmq}$ has radius $2$ if $b_i = 1$, and radius at most $2 - \Omega(\varepsilon)$ if $b_i = 0$, where $\varepsilon$ satisfies $n = \left(\sfrac{1}{\varepsilon}\right)^{\lfloor d / 4 \rfloor}$.
This gives us the following result.

\begin{restatable}{theorem}{ccmebrestatable} \label{ccmebthm}
    For $d < \left(\sfrac{1}{\varepsilon}\right)^{0.999}$, any $1$-pass streaming algorithm which yields a \linebreak$(1 + \eps)$-approximation for the MEB problem requires $\paren{\sfrac{1}{\eps}}^{\Omega(d)}$ words of space.
\end{restatable}

For the linear SVM problem, we show similar proofs, where Alice transforms her bitstring into points labeled $-1$ and Bob transforms his index into a point labeled $+1$. We then show similar properties about the separating hyperplane, giving us the following results.

\begin{restatable}{theorem}{ccsvmepsrestatable} \label{ccsvmepsthm}
    For $d < \left(\sfrac{1}{\varepsilon}\right)^{0.999}$, any $1$-pass streaming algorithm which yields a \linebreak$(1 + \eps)$-approximation for the linear SVM problem requires $\paren{\sfrac{1}{\eps}}^{\Omega(d)}$ words of space.
\end{restatable}

\begin{restatable}{theorem}{ccsvmgammarestatable} \label{ccsvmgammathm}
    For $d < \left(\sfrac{1}{\varepsilon}\right)^{0.999}$, any $1$-pass streaming algorithm which determines if a set of binary labeled points is $\gamma$-separable requires $\paren{\sfrac{1}{\gamma}}^{\Omega(d)}$ words of space.
\end{restatable}

A full explanation of these lower bounds are provided in Appendix~\ref{lbappendix}.

\bibliography{refs}

\begin{thebibliography}{10}

\bibitem{AHV04}
Pankaj~K. Agarwal, Sariel Har-Peled, and Kasturi~R. Varadarajan.
\newblock Approximating extent measures of points.
\newblock {\em J. ACM}, 51(4):606–635, jul 2004.
\newblock \href {https://doi.org/10.1145/1008731.1008736} {\path{doi:10.1145/1008731.1008736}}.

\bibitem{AS10}
Pankaj~K. Agarwal and R.~Sharathkumar.
\newblock Streaming algorithms for extent problems in high dimensions.
\newblock In {\em Proceedings of the Twenty-First Annual ACM-SIAM Symposium on Discrete Algorithms}, SODA '10, page 1481–1489, USA, 2010. Society for Industrial and Applied Mathematics.

\bibitem{AY07}
Pankaj~K. Agarwal and Hai Yu.
\newblock A space-optimal data-stream algorithm for coresets in the plane.
\newblock In {\em Proceedings of the Twenty-Third Annual Symposium on Computational Geometry}, SCG '07, page 1–10, New York, NY, USA, 2007. Association for Computing Machinery.
\newblock \href {https://doi.org/10.1145/1247069.1247071} {\path{doi:10.1145/1247069.1247071}}.

\bibitem{AG11}
Kook~Jin Ahn and Sudipto Guha.
\newblock Linear programming in the semi-streaming model with application to the maximum matching problem.
\newblock {\em Information and Computation}, 222:59--79, 2013.
\newblock 38th International Colloquium on Automata, Languages and Programming (ICALP 2011).
\newblock URL: \url{https://www.sciencedirect.com/science/article/pii/S0890540112001460}, \href {https://doi.org/10.1016/j.ic.2012.10.006} {\path{doi:10.1016/j.ic.2012.10.006}}.

\bibitem{AHLW16}
Yuqing Ai, Wei Hu, Yi~Li, and David~P. Woodruff.
\newblock New characterizations in turnstile streams with applications.
\newblock In {\em Proceedings of the 31st Conference on Computational Complexity}, CCC '16, Dagstuhl, DEU, 2016. Schloss Dagstuhl--Leibniz-Zentrum fuer Informatik.

\bibitem{AMS99}
Noga Alon, Yossi Matias, and Mario Szegedy.
\newblock The space complexity of approximating the frequency moments.
\newblock {\em Journal of Computer and System Sciences}, 58(1):137--147, 1999.
\newblock URL: \url{https://www.sciencedirect.com/science/article/pii/S0022000097915452}, \href {https://doi.org/10.1006/jcss.1997.1545} {\path{doi:10.1006/jcss.1997.1545}}.

\bibitem{AdFM17}
Sunil Arya, Guilherme~D. da~Fonseca, and David~M. Mount.
\newblock {Near-Optimal epsilon-Kernel Construction and Related Problems}.
\newblock In Boris Aronov and Matthew~J. Katz, editors, {\em 33rd International Symposium on Computational Geometry (SoCG 2017)}, volume~77 of {\em Leibniz International Proceedings in Informatics (LIPIcs)}, pages 10:1--10:15, Dagstuhl, Germany, 2017. Schloss Dagstuhl -- Leibniz-Zentrum f{\"u}r Informatik.
\newblock URL: \url{https://drops.dagstuhl.de/entities/document/10.4230/LIPIcs.SoCG.2017.10}, \href {https://doi.org/10.4230/LIPIcs.SoCG.2017.10} {\path{doi:10.4230/LIPIcs.SoCG.2017.10}}.

\bibitem{AKZ19}
Sepehr Assadi, Nikolai Karpov, and Qin Zhang.
\newblock Distributed and streaming linear programming in low dimensions.
\newblock In {\em Proceedings of the 38th ACM SIGMOD-SIGACT-SIGAI Symposium on Principles of Database Systems}, PODS '19, page 236–253, New York, NY, USA, 2019. Association for Computing Machinery.
\newblock \href {https://doi.org/10.1145/3294052.3319697} {\path{doi:10.1145/3294052.3319697}}.

\bibitem{BC03}
Mihai B{\u a}doiu and Kenneth~L. Clarkson.
\newblock Smaller core-sets for balls.
\newblock In {\em Proceedings of the Fourteenth Annual ACM-SIAM Symposium on Discrete Algorithms}, SODA '03, page 801–802, USA, 2003. Society for Industrial and Applied Mathematics.

\bibitem{BC08}
Mihai B{\u a}doiu and Kenneth~L. Clarkson.
\newblock Optimal core-sets for balls.
\newblock {\em Computational Geometry}, 40(1):14--22, 2008.
\newblock URL: \url{https://www.sciencedirect.com/science/article/pii/S0925772107000454}, \href {https://doi.org/10.1016/j.comgeo.2007.04.002} {\path{doi:10.1016/j.comgeo.2007.04.002}}.

\bibitem{BP22}
Yogesh~J. Bagul and Satish~K. Panchal.
\newblock Certain inequalities of {K}ober and {L}azarevic type.
\newblock {\em The Journal of the Indian Mathematical Society}, 89(1-2):01–07, Jan. 2022.
\newblock URL: \url{https://www.informaticsjournals.com/index.php/jims/article/view/20737}, \href {https://doi.org/10.18311/jims/2022/20737} {\path{doi:10.18311/jims/2022/20737}}.

\bibitem{BFKV96}
A.~Blum, A.~Frieze, R.~Kannan, and S.~Vempala.
\newblock A polynomial-time algorithm for learning noisy linear threshold functions.
\newblock In {\em Proceedings of 37th Conference on Foundations of Computer Science}, pages 330--338, 1996.
\newblock \href {https://doi.org/10.1109/SFCS.1996.548492} {\path{doi:10.1109/SFCS.1996.548492}}.

\bibitem{BJS03}
Yaroslav Bulatov, Sachin Jambawalikar, Piyush Kumar, and Saurabh Sethia.
\newblock Hand recognition using geometric classifiers.
\newblock In David Zhang and Anil~K. Jain, editors, {\em Biometric Authentication}, pages 753--759, Berlin, Heidelberg, 2004. Springer Berlin Heidelberg.

\bibitem{B98}
Christopher J.~C. Burges.
\newblock A tutorial on support vector machines for pattern recognition.
\newblock {\em Data Mining and Knowledge Discovery}, 2(2):121--167, 1998.
\newblock \href {https://doi.org/10.1023/A:1009715923555} {\path{doi:10.1023/A:1009715923555}}.

\bibitem{B94}
Tom Bylander.
\newblock Learning linear threshold functions in the presence of classification noise.
\newblock In {\em Proceedings of the Seventh Annual Conference on Computational Learning Theory}, COLT '94, page 340–347, New York, NY, USA, 1994. Association for Computing Machinery.
\newblock \href {https://doi.org/10.1145/180139.181176} {\path{doi:10.1145/180139.181176}}.

\bibitem{CR12}
Emmanuel Cand\`{e}s and Benjamin Recht.
\newblock Exact matrix completion via convex optimization.
\newblock {\em Communications of the ACM}, 55(6):111–119, June 2012.
\newblock \href {https://doi.org/10.1145/2184319.2184343} {\path{doi:10.1145/2184319.2184343}}.

\bibitem{CT10}
Emmanuel~J. Candes and Terence Tao.
\newblock The power of convex relaxation: Near-optimal matrix completion.
\newblock {\em IEEE Transactions on Information Theory}, 56(5):2053--2080, 2010.
\newblock \href {https://doi.org/10.1109/TIT.2010.2044061} {\path{doi:10.1109/TIT.2010.2044061}}.

\bibitem{CGLRML20}
Jair Cervantes, Farid Garcia-Lamont, Lisbeth Rodr{\'\i}guez-Mazahua, and Asdrubal Lopez.
\newblock A comprehensive survey on support vector machine classification: Applications, challenges and trends.
\newblock {\em Neurocomputing}, 408:189--215, 2020.
\newblock URL: \url{https://www.sciencedirect.com/science/article/pii/S0925231220307153}, \href {https://doi.org/10.1016/j.neucom.2019.10.118} {\path{doi:10.1016/j.neucom.2019.10.118}}.

\bibitem{C06}
Timothy~M. Chan.
\newblock Faster core-set constructions and data-stream algorithms in fixed dimensions.
\newblock {\em Computational Geometry}, 35(1):20--35, 2006.
\newblock Special Issue on the 20th ACM Symposium on Computational Geometry.
\newblock URL: \url{https://www.sciencedirect.com/science/article/pii/S0925772105000970}, \href {https://doi.org/10.1016/j.comgeo.2005.10.002} {\path{doi:10.1016/j.comgeo.2005.10.002}}.

\bibitem{CP11}
Timothy~M. Chan and Vinayak Pathak.
\newblock Streaming and dynamic algorithms for minimum enclosing balls in high dimensions.
\newblock In Frank Dehne, John Iacono, and J{\"o}rg-R{\"u}diger Sack, editors, {\em Algorithms and Data Structures}, pages 195--206, Berlin, Heidelberg, 2011. Springer Berlin Heidelberg.

\bibitem{C82}
M.~T. Chao.
\newblock A general purpose unequal probability sampling plan.
\newblock {\em Biometrika}, 69(3):653--656, 1982.
\newblock URL: \url{http://www.jstor.org/stable/2336002}.

\bibitem{C95}
Kenneth~L. Clarkson.
\newblock Las {V}egas algorithms for linear and integer programming when the dimension is small.
\newblock {\em J. ACM}, 42(2):488–499, mar 1995.
\newblock \href {https://doi.org/10.1145/201019.201036} {\path{doi:10.1145/201019.201036}}.

\bibitem{C10}
Kenneth~L. Clarkson.
\newblock Coresets, sparse greedy approximation, and the frank-wolfe algorithm.
\newblock {\em ACM Trans. Algorithms}, 6(4), sep 2010.
\newblock \href {https://doi.org/10.1145/1824777.1824783} {\path{doi:10.1145/1824777.1824783}}.

\bibitem{CHW12}
Kenneth~L. Clarkson, Elad Hazan, and David~P. Woodruff.
\newblock Sublinear optimization for machine learning.
\newblock {\em J. ACM}, 59(5), November 2012.
\newblock \href {https://doi.org/10.1145/2371656.2371658} {\path{doi:10.1145/2371656.2371658}}.

\bibitem{CM03}
Graham Cormode and S.~Muthukrishnan.
\newblock What's hot and what's not: tracking most frequent items dynamically.
\newblock In {\em Proceedings of the Twenty-Second ACM SIGMOD-SIGACT-SIGART Symposium on Principles of Database Systems}, PODS '03, page 296–306, New York, NY, USA, 2003. Association for Computing Machinery.
\newblock \href {https://doi.org/10.1145/773153.773182} {\path{doi:10.1145/773153.773182}}.

\bibitem{DV08}
John Dunagan and Santosh Vempala.
\newblock A simple polynomial-time rescaling algorithm for solving linear programs.
\newblock {\em Mathematical Programming}, 114(1):101--114, 2008.
\newblock \href {https://doi.org/10.1007/s10107-007-0095-7} {\path{doi:10.1007/s10107-007-0095-7}}.

\bibitem{GH16}
Dan Garber and Elad Hazan.
\newblock Sublinear time algorithms for approximate semidefinite programming.
\newblock {\em Mathematical Programming}, 158(1):329--361, 2016.
\newblock \href {https://doi.org/10.1007/s10107-015-0932-z} {\path{doi:10.1007/s10107-015-0932-z}}.

\bibitem{HW86}
David Haussler and Emo Welzl.
\newblock Epsilon-nets and simplex range queries.
\newblock In {\em Proceedings of the Second Annual Symposium on Computational Geometry}, SCG '86, page 61–71, New York, NY, USA, 1986. Association for Computing Machinery.
\newblock \href {https://doi.org/10.1145/10515.10522} {\path{doi:10.1145/10515.10522}}.

\bibitem{IMRUVY17}
Piotr Indyk, Sepideh Mahabadi, Ronitt Rubinfeld, Jonathan Ullman, Ali Vakilian, and Anak Yodpinyanee.
\newblock {Fractional Set Cover in the Streaming Model}.
\newblock In Klaus Jansen, Jos\'{e} D.~P. Rolim, David~P. Williamson, and Santosh~S. Vempala, editors, {\em Approximation, Randomization, and Combinatorial Optimization. Algorithms and Techniques (APPROX/RANDOM 2017)}, volume~81 of {\em Leibniz International Proceedings in Informatics (LIPIcs)}, pages 12:1--12:20, Dagstuhl, Germany, 2017. Schloss Dagstuhl -- Leibniz-Zentrum f{\"u}r Informatik.
\newblock URL: \url{https://drops.dagstuhl.de/entities/document/10.4230/LIPIcs.APPROX-RANDOM.2017.12}, \href {https://doi.org/10.4230/LIPIcs.APPROX-RANDOM.2017.12} {\path{doi:10.4230/LIPIcs.APPROX-RANDOM.2017.12}}.

\bibitem{IPW11}
Piotr Indyk, Eric Price, and David~P. Woodruff.
\newblock On the power of adaptivity in sparse recovery.
\newblock In {\em 2011 {IEEE} 52nd {A}nnual {S}ymposium on {F}oundations of {C}omputer {S}cience---{FOCS} 2011}, pages 285--294. IEEE Computer Soc., Los Alamitos, CA, 2011.
\newblock \href {https://doi.org/10.1109/FOCS.2011.83} {\path{doi:10.1109/FOCS.2011.83}}.

\bibitem{JST11}
Hossein Jowhari, Mert Sa\u{g}lam, and G\'{a}bor Tardos.
\newblock Tight bounds for lp samplers, finding duplicates in streams, and related problems.
\newblock In {\em Proceedings of the Thirtieth ACM SIGMOD-SIGACT-SIGART Symposium on Principles of Database Systems}, PODS '11, page 49–58, New York, NY, USA, 2011. Association for Computing Machinery.
\newblock \href {https://doi.org/10.1145/1989284.1989289} {\path{doi:10.1145/1989284.1989289}}.

\bibitem{KNW10}
Daniel~M. Kane, Jelani Nelson, and David~P. Woodruff.
\newblock An optimal algorithm for the distinct elements problem.
\newblock In {\em Proceedings of the Twenty-Ninth ACM SIGMOD-SIGACT-SIGART Symposium on Principles of Database Systems}, PODS '10, page 41–52, New York, NY, USA, 2010. Association for Computing Machinery.
\newblock \href {https://doi.org/10.1145/1807085.1807094} {\path{doi:10.1145/1807085.1807094}}.

\bibitem{KMY04}
Piyush Kumar, Joseph S.~B. Mitchell, and E.~Alper Yildirim.
\newblock Approximate minimum enclosing balls in high dimensions using core-sets.
\newblock {\em ACM J. Exp. Algorithmics}, 8:1.1–es, dec 2004.
\newblock \href {https://doi.org/10.1145/996546.996548} {\path{doi:10.1145/996546.996548}}.

\bibitem{LNW14}
Yi~Li, Huy~L. Nguyen, and David~P. Woodruff.
\newblock Turnstile streaming algorithms might as well be linear sketches.
\newblock In {\em Proceedings of the Forty-Sixth Annual ACM Symposium on Theory of Computing}, STOC '14, page 174–183, New York, NY, USA, 2014. Association for Computing Machinery.
\newblock \href {https://doi.org/10.1145/2591796.2591812} {\path{doi:10.1145/2591796.2591812}}.

\bibitem{LSZZZ23}
S.~Cliff Liu, Zhao Song, Hengjie Zhang, Lichen Zhang, and Tianyi Zhou.
\newblock {Space-Efficient Interior Point Method, with Applications to Linear Programming and Maximum Weight Bipartite Matching}.
\newblock In Kousha Etessami, Uriel Feige, and Gabriele Puppis, editors, {\em 50th International Colloquium on Automata, Languages, and Programming (ICALP 2023)}, volume 261 of {\em Leibniz International Proceedings in Informatics (LIPIcs)}, pages 88:1--88:14, Dagstuhl, Germany, 2023. Schloss Dagstuhl -- Leibniz-Zentrum f{\"u}r Informatik.
\newblock URL: \url{https://drops.dagstuhl.de/entities/document/10.4230/LIPIcs.ICALP.2023.88}, \href {https://doi.org/10.4230/LIPIcs.ICALP.2023.88} {\path{doi:10.4230/LIPIcs.ICALP.2023.88}}.

\bibitem{MVPV18}
Nantia Makrynioti, Nikolaos Vasiloglou, Emir Pasalic, and Vasilis Vassalos.
\newblock Modelling machine learning algorithms on relational data with datalog.
\newblock In {\em Proceedings of the Second Workshop on Data Management for End-To-End Machine Learning}, DEEM'18, New York, NY, USA, 2018. Association for Computing Machinery.
\newblock \href {https://doi.org/10.1145/3209889.3209893} {\path{doi:10.1145/3209889.3209893}}.

\bibitem{MSW96}
J.~Matou{\v s}ek, M.~Sharir, and E.~Welzl.
\newblock A subexponential bound for linear programming.
\newblock {\em Algorithmica}, 16(4):498--516, 1996.
\newblock \href {https://doi.org/10.1007/BF01940877} {\path{doi:10.1007/BF01940877}}.

\bibitem{NSWZ18}
Vasileios Nakos, Xiaofei Shi, David~P. Woodruff, and Hongyang Zhang.
\newblock Improved algorithms for adaptive compressed sensing.
\newblock In {\em 45th {I}nternational {C}olloquium on {A}utomata, {L}anguages, and {P}rogramming}, volume 107 of {\em LIPIcs. Leibniz Int. Proc. Inform.}, pages Art. No. 90, 14. Schloss Dagstuhl. Leibniz-Zent. Inform., Wadern, 2018.

\bibitem{NN04}
Frank Nielsen and Richard Nock.
\newblock Approximating smallest enclosing balls.
\newblock In Antonio Lagan{\'a}, Marina~L. Gavrilova, Vipin Kumar, Youngsong Mun, C.~J.~Kenneth Tan, and Osvaldo Gervasi, editors, {\em Computational Science and Its Applications -- ICCSA 2004}, pages 147--157, Berlin, Heidelberg, 2004. Springer Berlin Heidelberg.

\bibitem{N62}
A.~B. Novikoff.
\newblock On convergence proofs on perceptrons.
\newblock In {\em Proceedings of the Symposium on the Mathematical Theory of Automata}, volume~12, pages 615--622, New York, NY, USA, 1962. Polytechnic Institute of Brooklyn.

\bibitem{PW12}
Eric Price and David~P. Woodruff.
\newblock Lower bounds for adaptive sparse recovery.
\newblock In {\em Proceedings of the {T}wenty-{F}ourth {A}nnual {ACM}-{SIAM} {S}ymposium on {D}iscrete {A}lgorithms}, pages 652--663. SIAM, Philadelphia, PA, 2012.

\bibitem{R11}
Benjamin Recht.
\newblock A simpler approach to matrix completion.
\newblock {\em J. Mach. Learn. Res.}, 12(null):3413–3430, December 2011.

\bibitem{SW92}
Micha Sharir and Emo Welzl.
\newblock A combinatorial bound for linear programming and related problems.
\newblock In Alain Finkel and Matthias Jantzen, editors, {\em STACS 92}, pages 567--579, Berlin, Heidelberg, 1992. Springer Berlin Heidelberg.

\bibitem{SYZ23}
Zhao Song, Mingquan Ye, and Lichen Zhang.
\newblock Streaming semidefinite programs: ${O}(\sqrt{n})$ passes, small space and fast runtime, 2023.
\newblock URL: \url{https://arxiv.org/abs/2309.05135}, \href {https://arxiv.org/abs/2309.05135} {\path{arXiv:2309.05135}}.

\bibitem{SWYZ21}
Xiaoming Sun, David~P. Woodruff, Guang Yang, and Jialin Zhang.
\newblock Querying a matrix through matrix-vector products.
\newblock {\em ACM Trans. Algorithms}, 17(4), October 2021.
\newblock \href {https://doi.org/10.1145/3470566} {\path{doi:10.1145/3470566}}.

\bibitem{TKC05}
Ivor~W. Tsang, James~T. Kwok, and Pak-Ming Cheung.
\newblock Core vector machines: Fast svm training on very large data sets.
\newblock {\em J. Mach. Learn. Res.}, 6:363–392, dec 2005.

\bibitem{BLLSSSW21}
Jan van~den Brand, Yin~Tat Lee, Yang~P. Liu, Thatchaphol Saranurak, Aaron Sidford, Zhao Song, and Di~Wang.
\newblock Minimum cost flows, mdps, and $\ell_1$-regression in nearly linear time for dense instances.
\newblock In {\em Proceedings of the 53rd Annual ACM SIGACT Symposium on Theory of Computing}, STOC 2021, page 859–869, New York, NY, USA, 2021. Association for Computing Machinery.
\newblock \href {https://doi.org/10.1145/3406325.3451108} {\path{doi:10.1145/3406325.3451108}}.

\bibitem{VC71}
V.~N. Vapnik and A.~Ya. Chervonenkis.
\newblock On the uniform convergence of relative frequencies of events to their probabilities.
\newblock {\em Theory of Probability \& Its Applications}, 16(2):264--280, 1971.
\newblock \href {https://arxiv.org/abs/https://doi.org/10.1137/1116025} {\path{arXiv:https://doi.org/10.1137/1116025}}, \href {https://doi.org/10.1137/1116025} {\path{doi:10.1137/1116025}}.

\bibitem{VC15}
V.~N. Vapnik and A.~Ya. Chervonenkis.
\newblock On the uniform convergence of relative frequencies of events to their probabilities.
\newblock In Vladimir Vovk, Harris Papadopoulos, and Alexander Gammerman, editors, {\em Measures of Complexity: Festschrift for Alexey Chervonenkis}, pages 11--30. Springer International Publishing, Cham, 2015.
\newblock \href {https://doi.org/10.1007/978-3-319-21852-6_3} {\path{doi:10.1007/978-3-319-21852-6_3}}.

\bibitem{W14}
David~P. Woodruff.
\newblock Sketching as a tool for numerical linear algebra.
\newblock {\em Found. Trends Theor. Comput. Sci.}, 10(1–2):1–157, October 2014.
\newblock \href {https://doi.org/10.1561/0400000060} {\path{doi:10.1561/0400000060}}.

\bibitem{XJRN02}
Eric Xing, Michael Jordan, Stuart~J Russell, and Andrew Ng.
\newblock Distance metric learning with application to clustering with side-information.
\newblock In S.~Becker, S.~Thrun, and K.~Obermayer, editors, {\em Advances in Neural Information Processing Systems}, volume~15. MIT Press, 2002.
\newblock URL: \url{https://proceedings.neurips.cc/paper_files/paper/2002/file/c3e4035af2a1cde9f21e1ae1951ac80b-Paper.pdf}.

\bibitem{YT89}
Yinyu Ye and Edison Tse.
\newblock An extension of {K}armarkar's projective algorithm for convex quadratic programming.
\newblock {\em Mathematical Programming}, 44(1):157--179, 1989.
\newblock \href {https://doi.org/10.1007/BF01587086} {\path{doi:10.1007/BF01587086}}.

\bibitem{Y08}
E.~Alper Yildirim.
\newblock Two algorithms for the minimum enclosing ball problem.
\newblock {\em SIAM J. on Optimization}, 19(3):1368–1391, nov 2008.
\newblock \href {https://doi.org/10.1137/070690419} {\path{doi:10.1137/070690419}}.

\bibitem{Z11}
Hamid Zarrabi-Zadeh.
\newblock An almost space-optimal streaming algorithm for coresets in fixed dimensions.
\newblock {\em Algorithmica}, 60(1):46--59, 2011.
\newblock \href {https://doi.org/10.1007/s00453-010-9392-2} {\path{doi:10.1007/s00453-010-9392-2}}.

\bibitem{ZC06}
Hamid Zarrabi{-}Zadeh and Timothy~M. Chan.
\newblock A simple streaming algorithm for minimum enclosing balls.
\newblock In {\em Proceedings of the 18th Annual Canadian Conference on Computational Geometry, {CCCG} 2006, August 14-16, 2006, Queen's University, Ontario, Canada}, 2006.
\newblock URL: \url{http://www.cs.queensu.ca/cccg/papers/cccg36.pdf}.

\end{thebibliography}

\appendix

\section{Algorithm Specifics}\label{algorithmappendix}

In this section, we explain some intricacies of our algorithm, provide pseudocode for it, and prove its correctness and iteration count.

First, we explain our lattice construction and some discussion on it. In our construction, we start by taking any point to use as the center of our metric \epsnet{}. For example, in the multipass streaming model, this can be the first point. We call this point $\bmp_c$. Now, we can shift our input points to put $\bmp_c$ at the origin. This will also give us the useful fact that now any other point's norm is defined as its distance to $\bmp_c$. We can use this point to create a metric \epsnet{} by allowing, for each $d$ dimensions, the values $-1 + i \frac{\eps}{\sqrt{d}}$ where $i$ ranges in $\brak{0, 2 \frac{\sqrt{d}}{\eps}}$. 
It is clear to see that this creates a metric \epsnet{} as it is is a lattice covering the unit $d$-cube (which contains the unit $d$-sphere), where the distance between net points $\bme_i$ and $\bme_{i+1}$, where $i$ denotes the value for each coordinate, is $\eps$. 
Thus, any point that is scaled to the unit $d$-sphere will be at most $\frac{\eps}{2}$ away from a point in our metric \epsnet{}. The size of this metric \epsnet{} is $\paren{1 + \frac{2\sqrt{d}}{\eps}}^d$. This net is sub-optimal in size. For example, one can save a factor of $\sqrt{d}$ by covering the unit $\ell_2$-ball. However, in our algorithm we retain complexity logarithmic in the net size, and so in the high accuracy regime where $d < (\sfrac{1}{\eps})^{0.999}$, these are asymptotically equivalent as $O \paren{\log \paren{\sfrac{\sqrt{d}}{\eps}}} \in O \paren{ \log \paren{\sfrac{1}{\eps}}}$. We use the lattice net because it can be described succinctly, thus requiring very small additional complexity to store the net itself. Further, it is easy to calculate the closest metric \epsnet{} point $\bmq_j$ for an arbitrary point $\bmp_i$.

As noted in the body text, snapping the input to this \epsnet{} brings a complication to our weighted sampling procedure. Because our linear sketch is formed by snapping each input point $\bmp \in P$ to a point $\bmq \in Q$, where $\abs{Q} = N \ll \abs{P} = n$, we can have distinct original points $\bmp \neq \bmp'$ that are snapped to the same point $\bmq$. Since our weight function deals with weights for the snapped points $\bmq$, and there is no easy way to recognize duplicates in large (and possibly distributed) input sets, sampling using reservoir sampling over $P$ will favor these duplicated points. Thus, we utilize $\ell_0$ estimators and $\ell_0$ samplers in order to sample proportionally to the number of distinct snapped points $\bmq$. As these estimators and samplers are linear sketches, they can be implemented in our algorithm. We use the benefit of the multiple linear sketch model here, as the weight function we use changes throughout the algorithm. Thus, we are able to create different linear sketches each iteration of our algorithm. The tradeoff is that they come with approximations and error probabilities, however we show that the distribution we sample from by using approximate weights is within a small variation distance of the distribution we would obtain from using true weights. As such, we retain the same analysis and approximate bounds.

Now, for the algorithm procedure. In each iteration, the goal of the algorithm is to sample a small subset of the snapped points that contain a basis for the optimal solution to the problem while keeping a low space complexity.
This goal is achieved by sampling a \munet{} with respect to $w(\cdot)$ from the snapped points using Lemma~\ref{mdrawslemma}. As we explain previously, this sampling cannot be done directly from the snapped points $Q$ because of overlaps, and so we need to utilize $\ell_0$ estimators and samplers. Since each point can increase in weight at most once in each iteration, we have a number of discrete weight classes, for which we create estimators and samplers. This allows us to randomly sample each unique snapped point with probability proportional to its weight by randomly selecting a weight class in accordance to the class' total estimated weight, and then using its sampler to uniformly sample a point from that weight class.

Here we present pseudocode for the algorithm in Algorithm~\ref{MEBalgo}, which uses Algorithms~\ref{Samplehelperalgo} and~\ref{Weighthelperalgo} as helper functions to sample a \munet{} and to check the success of an iteration.
Next, we show that our algorithm correctly computes a valid $(1 + O(\eps))$-approximation to the solution of the LP-type problem, $f(P)$ at the end of its computation.

\begin{algorithm}[H]
    \DontPrintSemicolon
    \SetKwInOut{Input}{input}\SetKwInOut{Output}{output}
    \SetKwFunction{SampleFn}{SampleMPoints}
    \SetKwFunction{ViolatorFn}{CheckViolatorsWeight}
    \Input{A stream $P$ of $n$ points $\bmp \in \cbrak{-\Delta, -\Delta + 1, \hdots, \Delta}^d$}
    \Output{A $1 + O(\eps)$ approximation for $f(P)$}
    \BlankLine
    Take a point $\bmp_1 \in P$, as the origin\;
    Create the metric \epsnet{} $\cbrak{(v_1, v_2, \hdots, v_d) \in \bbR^d \mid v_i = -1 + j\frac{\eps}{\sqrt{d}} \text{ where } j \in \brak{0, 2\frac{\sqrt{d}}{\eps}}}$\;
    $r_m \leftarrow \max_{\bmp_i \in P} \norm{\bmp_i}$\;
    Define $N := \ceil{\log_{1 + \eps} r_m} \paren{1 + \frac{2 \sqrt{d}}{\eps}}^d$\;
    Choose $s \leq \ln N$\;
    Define $\mu := \frac{1}{10\nu N^{\sfrac{1}{s}}}$ where $\nu$ is the combinatorial dimension of the LP-type problem.\;
    \Repeat{$V = \emptyset$}{
        $B \leftarrow$ \SampleFn{P, $\mu$, $w(\cdot)$, current iteration number $t$}\;
        \If{\ViolatorFn{P, $\mu$, $w(\cdot)$, current iteration number $t$, $f(B)$}}{
            Denote iteration successful\;
            Set $w(\bmq) = w(\bmq) \cdot N^{\sfrac{1}{s}}$ for each violator $\bmq \in V$\;
        }
    }
    Let $f(B)$ be the last solution\;
    \Return{$(1 + O(\eps))$ approximation to $f(B)$}
    \caption{An Algorithm for LP-Type Problems}
    \label{MEBalgo}
\end{algorithm}

\begin{algorithm}[H]
    \DontPrintSemicolon
    \SetKwInOut{Input}{input}\SetKwInOut{Output}{output}
    \SetKwFor{RepTimes}{repeat}{times}{end}
    \Input{A stream $P$ of $n$ points $\bmp \in \cbrak{-\Delta, -\Delta + 1, \hdots, \Delta}^d$, a parameter $\mu$, weight function $w(\cdot)$, and integer $t$}
    \Output{A set $B$ of unique points in $m := \max\paren{\frac{8 \lambda}{\mu} \log \frac{8 \lambda}{\mu}, \frac{4}{\mu} \log \frac{2}{\frac{1}{4}}}$ samples drawn with replacement according to $w(\cdot)$}
    \BlankLine
    Define $m := \max\paren{\frac{8 \lambda}{\mu} \log \frac{8 \lambda}{\mu}, \frac{4}{\mu} \log \frac{2}{\frac{1}{4}}}$\;
    Create $t$ $(1 \pm \frac{1}{m^{\sfrac{3}{2}}})$-approximation $\ell_0$ estimators with error probability $\frac{1}{\polyof{N}}$ and $t$ $\ell_0$ samplers with error probability $\frac{1}{\polyof{N}}$, (one for each weight class), where the underlying vectors are of size $N$ and indexed by points $\bmq\in Q$.\;
    \ForEach{$\bmp_i \in P$}{
        $\bmq_j \leftarrow \bmp_i$ snapped to nearest metric \epsnet{} direction and rounded down\;
        Increment $w(\bmq_j)$\textsuperscript{th} estimator and sampler by $1$ in the $\bmq_j$\textsuperscript{th} position\;
    }
    $f_i \leftarrow$ estimate for each weight class estimator $i \in [t]$\;
    $\bmp_i \leftarrow \frac{f_i N^{\sfrac{i}{s}}}{\sum_{j = 1}^{t} f_j N^{\sfrac{j}{s}}}$\;
    $B \leftarrow \emptyset$\;
    \RepTimes{$m$}{
        Pick $j \in [t]$ randomly according to probabilities $\bmp_i$\;
        $\bmq \leftarrow$ sample point from $j$\textsuperscript{th} sampler\;
        $B \leftarrow B \cup \cbrak{\bmq}$
    }
    \Return{$B$}
    \caption{\texttt{SampleMPoints}}
    \label{Samplehelperalgo}
\end{algorithm}

\begin{algorithm}[H]
    \DontPrintSemicolon
    \SetKwInOut{Input}{input}\SetKwInOut{Output}{output}
    \Input{A stream $P$ of $n$ points $\bmp \in \cbrak{-\Delta, -\Delta + 1, \hdots, \Delta}^d$, a parameter $\mu$, weight function $w(\cdot)$, integer $t$, and $f(B)$}
    \Output{A boolean representing if the weight of violators for $f(B)$ is small enough}
    \BlankLine
    Create two sets of $t$ $(1 \pm \frac{1}{4})$-approximation $\ell_0$ estimators with error probability $\frac{1}{\polyof{N}}$, (two for each weight class), where the underlying vectors are of size $N$ and indexed by points $\bmq\in Q$\;
    \ForEach{$\bmp_i \in P$}{
        $\bmq_j \leftarrow \bmp_i$ snapped to nearest metric \epsnet{} direction and rounded down\;
        Increment $w(\bmq_j)$\textsuperscript{th} estimator in the first set by $1$ in the $\bmq_j$\textsuperscript{th} position\;
        \If(\tcp*[f]{$\bmq_j$ is a violator}){$f\paren{B \cup \cbrak{\bmq_j}} > f(B)$}{
            Increment $w(\bmq_j)$\textsuperscript{th} estimator in the second set by $1$ in the $\bmq_j$\textsuperscript{th} position\;
        }
    }
    $f_i' \leftarrow$ estimate for each estimator $i \in [t]$ in the first set\;
    $f_i'' \leftarrow$ estimate for each estimator $i \in [t]$ in the second set\;
    \Return{$\frac{1}{(1+\frac{1}{4})} \sum_i f_i'' N^{\sfrac{i}{s}} \leq \mu \cdot \frac{1}{(1-\frac{1}{4})} \sum_i f_i' N^{\sfrac{i}{s}}$}
    \caption{\texttt{CheckViolatorsWeight}}
    \label{Weighthelperalgo}
\end{algorithm}

\begin{claim}
    Given an LP-type problem $(S, f)$ and a set of input points $P$, Algorithm~\ref{MEBalgo} correctly returns a $(1 + O(\eps))$-approximation to $f(P)$.
\end{claim}
\begin{proof}
    The algorithm only returns when a computed solution $f(B)$ has no violators. Thus, for any point $\bmq_j \notin B$, $f\paren{B \cup \cbrak{\bmq_j}} \leq f(B)$. Further, since $B \subseteq B \cup \cbrak{\bmq_j}$, by the monotonicity property for LP-type problems, we have that $f(B) \leq f\paren{B \cup \cbrak{\bmq_j}}$. Such, we have that $f(B) = f\paren{B \cup \cbrak{\bmq_j}}$.
    
    We can further this by using the locality property for LP-type problems, so that for any $\bmq_k \notin B$, $f(B) = f\paren{B \cup \cbrak{\bmq_j} \cup \cbrak{\bmq_k}}$. Proceeding in an inductive setting for all $\bmq_j \in Q \setminus B$, we will have that 
    $f(B) = f\paren{B \cup \paren{Q \setminus B}} = f(Q)$.

    Finally, we use property~\ref{P3}, in order to conclude that our result $f(B)$ is a $(1 + O(\eps))$-approximation to $f(P)$.
\end{proof}

\subsection{Algorithm Iteration Count}\label{iterationsubsection}
In this section, we bound the total number of iterations our algorithm takes at $O(\nu s)$ where $\nu$ is the combinatorial dimension of the LP-type problem, and $s$ is our parameter with range $\brak{1, \log N}$. To do this, we first bound the number of successful iterations our algorithm takes, which we then use to achieve a total bound on the number of iterations.

\begin{restatable}{claim}{succprobrestatable} \label{2/3succ}
    Each iteration of Algorithm~\ref{MEBalgo} is denoted successful with probability at least $2/3$.
\end{restatable}
\begin{proof}
    The intuition behind this claim is that when points are sampled in accordance to their weights, we can use Lemma~\ref{mdrawslemma} to state that our sample is a \munet{} with probability $\sfrac{3}{4}$. While our sampling is not exact and based on $\ell_0$ estimators and samplers, we show using a TV distance argument that we retain this $\sfrac{3}{4}$ probability. Using this, we show that our condition for success is fulfilled with probability $\frac{3}{4} - \frac{1}{\polyof{N}}$, where the error is due to the error of our estimators, which still gives an overall $\sfrac{2}{3}$ success probability. A full proof is provided in Appendix~\ref{successfuliterationappendix}.
\end{proof}

We now note that the weight function $w(\cdot)$ is only updated when there is an iteration that is deemed successful. By Claim~\ref{2/3succ} each iteration is deemed successful with probability at least $\sfrac{2}{3}$. By the Chernoff bound, if the algorithm takes $t$ iterations, then with probability at least $1 - e^{-\Omega(t)}$, at least $\sfrac{t}{2}$ of these iterations are denoted as successful. It follows that bounding the number of successful iterations using the weight function gives us a bound on the total number of iterations.

\begin{restatable}{claim}{numiterrestatable} \label{numiterclaim}
    Algorithm~\ref{MEBalgo} takes $O(\nu s)$ iterations with probability at least $1 - \frac{1}{\polyof{N}} - e^{- \Omega(t)}$.
 \end{restatable}
\begin{proof}
    The proof follows by upper and lower bounding the weight of the set $Q$ after the $t$\textsuperscript{th} iteration denoted successful, which can then be used together to bound $t$ at $O(\nu s)$ with high probability. The proof is provided in Appendix~\ref{numiterappendix}.
\end{proof}

\subsection{Proof of Claim~\ref{2/3succ}} \label{successfuliterationappendix}
\succprobrestatable*

\begin{proof}
    The first step towards showing this is getting a probability bound on how many iterations will have the sample be a \munet{}. By Lemma~\ref{mdrawslemma}, a sample obtained from $m := \max\paren{\frac{8\lambda}{\mu} \log \frac{8\lambda}{\mu}, \frac{4}{\mu} \log \frac{2}{\frac{1}{4}}}$ draws with probability proportional to the points' weights will result in a \munet{} with respect to $w(\cdot)$ with probability at least $3/4$. However, we have to account for the approximations and error probabilities of our $\ell_0$ estimators and samplers, as we use them to get our draws.

    To account for this, we will look at the TV distance between the distributions where the $m$ drawn samples are taken with probabilities according to the actual weights, and with probabilities according to our estimates.
    
    \begin{definition}
        The Total Variation (TV) distance is a metric on the space of probability distributions over a discrete set $X$. If $\mcP$ and $\mcQ$ are distributions over $X$, the TV distance between $\mcP$ and $\mcQ$ is given by
        \[ \TV(\mcP, \mcQ) = \frac{1}{2} \sum_{x \in X} |\Pr_\mcP(x) - \Pr_\mcQ(x)|.  \]
        The TV distance is also an integral probability metric with respect to the class of functions with co-domain $\{0, 1\}$. That is, 
        \[ \TV(\mcP, \mcQ) = \sup\left\{
        \left| \Exp{x \sim \mcP}{f(x)} - \Exp{x \sim \mcQ}{f(x)} \right| : f : X \to \{0, 1\} \right\}. \]
    \end{definition}
    
    Let $\mcP$ be the actual distribution, and $\mcQ$ be our estimated distribution. Then, looking at the TV distance, we have that
    \begin{align}
        \TV(\mcP, \mcQ) &= \frac{1}{2} \sum_{B} \abs{\PrP(B) - \PrQ(B)} \nonumber\\
        &= \frac{1}{2} \sum_{B} \PrQ(B) \abs{\frac{\PrP(B)}{\PrQ(B)} - 1} \label{tveq}.
    \end{align}
    We can now bound the quotient of these two probabilites. For simplicity and generality, we will refer to the approximation ratio of the $\ell_0$ samplers as $(1 + \zeta)$.

    \begin{lemma} \label{tvlemma}
        $\frac{\PrP(B)}{\PrQ(B)} \leq \paren{1 + \frac{1}{\polyof{N}}}^m \paren{\frac{1 + \zeta}{1 - \zeta}}^m$
    \end{lemma}
    \begin{proof}
        We proceed by induction on $m$.

        \subparagraph*{Base Case}
        Consider a singleton sample $B_1$, and let $K_1$ be the event of $B_1$'s weight class being chosen. Then we get the quotient
        \begin{align*}
            \frac{\PrP(B_1)}{\PrQ(B_1)} &= \frac{\PrP(B_1 \mid K_1)}{\PrQ(B_1 \mid K_1)} \frac{\PrP(K_1)}{\PrQ(K_1)}\\
            &\leq \paren{1 + \frac{1}{\polyof{N}}}\frac{\PrP(K_1)}{\PrQ(K_1)} && \text{by errors of the $\ell_0$ samplers and estimators}\\
            &\leq \paren{1 + \frac{1}{\polyof{N}}} \paren{\frac{1 + \zeta}{1 - \zeta}} && \text{by error of the $\ell_0$ estimators}
        \end{align*}
        as desired.

        \subparagraph*{Induction Step}
        Consider now the sample $B$ of $m$ sample points $B_1, \hdots, B_m$, and let $K_i$ be the event of $S_i$'s weight class being chosen. Then we get the quotient
        \begin{align*}
            \frac{\PrP(B)}{\PrQ(B)} &= \frac{\PrP(B_m)}{\PrQ(B_m)} \frac{\PrP(B_1, \hdots, B_{m-1})}{\PrQ(B_1, \hdots, B_{m-1})}\\
            &\leq \frac{\PrP(B_m)}{\PrQ(B_m)} \paren{1 + \frac{1}{\polyof{N}}}^{m-1} \paren{\frac{1 + \zeta}{1 - \zeta}}^{m-1} && \text{by the induction hypothesis}.
        \end{align*}
        We can now notice that $\frac{\PrP(B_m)}{\PrQ(B_m)} = \frac{\PrP(B_1)}{\PrQ(B_1)}$ since the samples are drawn independently with replacement. Thus, we get that
        \[
            \frac{\PrP(B)}{\PrQ(B)} = \paren{1 + \frac{1}{\polyof{N}}}^m \paren{\frac{1 + \zeta}{1 - \zeta}}^m
        \]
        completing our induction proof.
    \end{proof}

    Now, we can plug this back into Equation~\ref{tveq}, and use the fact that the sum of all probabilities of a distribution is $1$, to get that
    \[
        \TV(\mcP, \mcQ) \leq \frac{1}{2}\paren{ \paren{1 + \frac{1}{\polyof{N}}}^m \paren{\frac{1 + \zeta}{1 - \zeta}}^m - 1}.
    \]
    Using our $\ell_0$ sampler approximation ratios of $\paren{1 + \frac{1}{m^{3/2}}}$, we get that this TV distance converges to $0$, which means that our distribution will be sufficiently close to the actual distribution, and so we will retain that each sample will be a \munet{} with probability at least $3/4$. We will now case on this event when our sample $B$ is a \munet{}.

    Let $f(B)$ be the solution of the sample $B$. First, we claim that for the set of violators $V$, we have that $w(V) \leq \mu \cdot w(Q)$. Assume for the sake of contradiction that $w(V) > \mu \cdot w(Q)$. Since $V \in \mcR$ and $B$ is a $\mu$-net with respect to $w(\cdot)$, this implies that $B \cap V \neq \emptyset$. So, consider a point $\bme \in B \cap V$. Since $\bme \in V$, $f(B \cup \cbrak{\bme}) > f(B)$. However, since $\bme \in B$, $B \cup \cbrak{\bme} = B$, and such it cannot be that $f(B \cup \cbrak{\bme}) > f(B)$. Thus, due to this contradiction, we have that $w(V) \leq \mu \cdot w(S)$.

    Now, we want to show that $\frac{1}{(1+\eta)} \sum_i f_i'' N^{i/s} \leq \mu \cdot \frac{1}{(1-\eta)} \sum_i f_i' N^{i/s}$, as that is when we denote an iteration successful. First, we define $f_{i,V}$ as the number of violators in the $i$\textsuperscript{th} weight class, and similarly $f_{i,Q}$ as the number of snapped points in the $i$\textsuperscript{th} weight class. Using these definitions, we can rewrite $w(V)$ and $w(Q)$, separating by weight class, as $\sum_i f_{i,V} N^{i/s}$ and $\sum_i f_{i,Q} N^{i/s}$ respectively. Plugging these new expressions in, we have that
    \begin{equation} \label{eq1}
        \sum_i f_{i,V} N^{i/s} \leq \mu \sum_i f_{i,Q} N^{i/s}.
    \end{equation}
    By the definition of our estimators, for any $i$ we have that
        \[
            (1 - \eta) f_{i, Q} \leq f_i' \leq (1 + \eta) f_{i, Q}.
        \]
        
        So, we have that $f_{i,Q} \leq \frac{1}{(1-\eta)} f_i'$, and so multiplying by the weight of the class and summing over all weight classes we have that
        \[
            \sum_i f_{i,Q} N^{i/s} \leq \frac{1}{(1-\eta)} \sum_i f_i' N^{i/s}.
        \]

        Similarly, as we have that
        \[
            (1 - \eta) f_{i, V} \leq f_i'' \leq (1 + \eta) f_{i, V},
        \]
        we have that
        \[
            \frac{1}{(1+\eta)} \sum_i f_i'' N^{i/s} \leq \sum_i f_{i,V} N^{i/s}.
        \]

        Plugging these into (\ref{eq1}), we get that
        \[
            \frac{1}{(1+\eta)} \sum_i f_i'' N^{i/s} \leq \mu \cdot \frac{1}{(1-\eta)} \sum_i f_i' N^{i/s},
        \]
        with probability at least $1 - \frac{1}{\polyof{N}}$, which comes from the error probability of the estimators.

        So we know that each iteration will be denoted a successful iteration with probability at least $\frac{3}{4} - \frac{1}{\polyof{N}}$. Given a sufficiently large $N$ such that $\frac{1}{\polyof{N}} \leq  1/12$, we have that each iteration will be denoted a successful iteration with probability at least $2/3$.
\end{proof}

\subsection{Proof of Claim~\ref{numiterclaim}} \label{numiterappendix}
\numiterrestatable*

\begin{proof}
    For the sake of simplicity and generality, with we will refer of the approximation ratio of the $\ell_0$ samplers as $(1 + \eta)$.

    Note that since for all $\bmq \in Q$ $w_0(\bmq) = 1$, $w_0(Q) \leq N$.

    Now, we claim that for any integer $t \geq 1$ after the $t$\textsuperscript{th} iteration denoted successful,
    \begin{equation} \label{2sideineq}
        N^{t/\nu s} \leq w_t(Q) \leq e^{\frac{(1 + \eta)^2 t}{10 (1 - \eta)^2 \nu}} \cdot N.
    \end{equation}

    \begin{lemma}
        For any integer $t \geq 1$, after the $t$\textsuperscript{th} iteration denoted successful,
        \[
            {N^{t/\nu s} \leq w_t(Q)}.
        \]
    \end{lemma}
    \begin{proof}
        Consider an arbitrary basis $B^*$ of $Q$, defined by $k$ points for some $k \leq \nu$. Since $B^* \subseteq Q$, $w_t(B^*) \leq w_t(Q)$. Thus, it suffices to show that $N^{t/\nu s} \leq w_t(B^*)$.
        
        Now, observe that for the set of violators $V$ for an iteration, if $V \neq \emptyset$, then $V \cap B^* \neq \emptyset$.  This is because if we assume for contradiction that $V \cap B^* = \emptyset$, then for any point $\bme \in B^*$, $\bme \notin V$, meaning that $f(B \cup \cbrak{\bme}) \leq f(B)$. Further, from the monotonicity property, $f(B \cup \cbrak{\bme}) \geq f(B)$, and so $f(B \cup \cbrak{\bme}) = f(B)$.
        Inductively repeating this for all points in $B^*$ using the locality property, we get that ${f(B \cup B^*) = f(B)}$. Further, by the monotonicity property, we have that ${f(B^*) \leq f(B \cup B^*) \leq f(Q)}$. But since $B^*$ is a basis for $Q$, it must be that $f(B^*) = f(Q)$. Connecting these all together, we get that $f(B) = f(Q)$.

        Finally, to achieve our contradiction, we can notice that in this case it cannot be that $B$ has any violating points $\bme'$, since by monotonicity $f(B \cup \cbrak{\bme'}) \leq f(Q) = f(B)$. So, we have a contradiction to the fact that $V \neq \emptyset$, and thus $V \cap B^* \neq \emptyset$.

        Now for each $i \in [t]$, define $B_i$ as the sample in the $i$\textsuperscript{th} iteration denoted successful. For any $l \in [k]$, let $a_l$ be the number of samples $B_i$ that are violated by $\bmq_l \in B^*$, i.e., 
        \[
            a_l = \left| \left\{ 
                i \in [t] \mid f(B_i \cup \cbrak{\bmq_l}) > f(B_i)
            \right\} \right|.
        \]
        Since in each iteration denoted successful, there have been some violators (as the algorithm did not return after it), we know that $V \cap B^* \neq \emptyset$ for these iterations, and so there must exist at least one $\bmq_l$ which violated $B_i$ for each $i \in [t]$. Thus, we have that $\sum_{l=1}^k a_l \geq t$.

        Now, looking at $B^*$, we can see that
        \begin{align*}
            w_t(B^*) &= \sum_{l = 1}^k w_t(\bmq_l)\\
            &= \sum_{l = 1}^k \paren{N^{1 / s}}^{a_l}\\
            &\geq k \paren{N^{\sfrac{1}{s}}}^{\sum_{l = 1}^k a_l/k} && \text{by Jensen's inequality}\\
            &\geq k \paren{N^{\sfrac{1}{s}}}^{t/k}\\
            &\geq N^{t/\nu s} && \text{as $k \leq \nu$}
        \end{align*}
        proving the lemma.
    \end{proof}

    \begin{lemma}
        For any integer $t \geq 1$, after the $t$\textsuperscript{th} iteration denoted successful,
        \[
            {w_t(Q) \leq e^{\frac{(1 + \eta)^2 t}{10 (1 - \eta)^2 \nu}} \cdot N}.
        \]
    \end{lemma}
    \begin{proof}
        First, by how the weights are set, we have that
        \begin{equation} \label{eq2}
            w_{t + 1}(Q) = w_t(SQ) + (N^{\sfrac{1}{s}} - 1) \cdot w_t(V) \leq w_t(Q) + N^{\sfrac{1}{s}} \cdot w_t(V).
        \end{equation}
        We also know from our definition of a successful iteration that
        \begin{equation} \label{eq3}
            \frac{1}{(1+\eta)} \sum_i f_{i,t}'' N^{i/s} \leq \mu \cdot \frac{1}{(1-\eta)} \sum_i f_{i,t}' N^{i/s}
        \end{equation}
        where the estimate $f_i'$ on the $t$\textsuperscript{th} iteration is denoted $f_{i,t}'$ and similarly for $f_{i,t}''$.
        We rewrite $w_t(Q)$ and $w_t(V)$ in this form, as $\sum_i f_{i,t,Q} N^{i/s}$ and $\sum_i f_{i,t,V} N^{i/s}$ respectively.

        By the definition of our estimators, for any $i$ we have that
        \[
            (1 - \eta) f_{i, t, Q} \leq f_{i,t}' \leq (1 + \eta) f_{i, t, Q}.
        \]
        
        So, we have that $\frac{1}{(1 - \eta)} f_{i,t}' \leq \frac{(1 + \eta)}{(1-\eta)} f_{i, t, Q}$, and so multiplying by the weight of the class and summing over all weight classes we have that
        \[
            \frac{1}{(1-\eta)} \sum_i f_{i,t}' N^{i/s} \leq \frac{(1 + \eta)}{(1-\eta)} \sum_i f_{i,t, Q} N^{i/s} = \frac{(1 + \eta)}{(1-\eta)} w_t(Q).
        \]
        Similarly, as we have that
        \[
            (1 - \eta) f_{i, t, V} \leq f_{i,t}'' \leq (1 + \eta) f_{i, t, V},
        \]
        we have that
        \[
            \frac{(1-\eta)}{(1+\eta)} w_t(V) = \frac{(1-\eta)}{(1+\eta)} \sum_i f_{i,t,V} N^{i/s} \leq \frac{1}{(1+\eta)} \sum_i f_{i,t}'' N^{i/s}.
        \]

        Plugging these into (\ref{eq3}), we get that
        \[
            \frac{(1-\eta)}{(1+\eta)} w_t(V) \leq \mu \cdot \frac{(1 + \eta)}{(1-\eta)} w_t(Q)
        \]
        with probability at least $\frac{1}{\polyof{N}}$, which comes from the error probability of the estimators. This can further be rearranged into
        \[
            w_t(V) \leq \mu \cdot \frac{(1 + \eta)^2}{(1-\eta)^2} w_t(Q).
        \]
        Now, plugging this into (\ref{eq2}), we get that
        \begin{align*}
            w_{t}(Q) &\leq w_{t - 1}(Q) + N^{\sfrac{1}{s}} \cdot \mu \cdot \frac{(1 + \eta)^2}{(1-\eta)^2} w_{t - 1}(Q)\\
            &= \paren{1 + N^{\sfrac{1}{s}} \cdot \mu \cdot \frac{(1 + \eta)^2}{(1-\eta)^2}} w_{t - 1}(Q)\\
            &= \paren{1 + N^{\sfrac{1}{s}} \cdot \frac{1}{10 \cdot \nu \cdot N^{\sfrac{1}{s}}} \cdot \frac{(1 + \eta)^2}{(1-\eta)^2}} w_{t - 1}(Q)\\
            &= \paren{1 + \frac{1}{10 \cdot \nu} \cdot \frac{(1 + \eta)^2}{(1-\eta)^2}} w_{t - 1}(Q).
        \end{align*}
        We can further unroll this inequality until we get
        \[
            w_t(Q) \leq \paren{1 + \frac{(1 + \eta)^2}{10 \nu (1-\eta)^2}}^t w_0(Q) \leq e^{\frac{(1 + \eta)^2 t}{10 \nu (1 - \eta)^2}} \cdot N
        \]
        with probability at least $1 - \frac{1}{\polyof{N}}$, proving the lemma.
    \end{proof}

    Connecting both sides of (\ref{2sideineq}), we get $N^{t/\nu s} \leq e^{\frac{(1 + \eta)^2 t}{10 (1 - \eta)^2 \nu}} \cdot N$, which can be rearranged to get 
    \[
        \frac{t}{\nu} \leq \frac{10 (1 - \eta)^2}{10(1-\eta)^2 - (1+\eta)^2} s.
    \]
    Now, taking $\eta = 1/4$, we can simplify this as $\frac{t}{\nu} \leq \frac{3}{2} s$, bounding the number of successful iterations $t$ at $\frac{3}{2}\nu s$. Thus, the total number of iterations the algorithm takes is bounded by $O(\nu s)$ with probability at least $1 - \frac{1}{\polyof{N}} - e^{- \Omega(t)}$.
\end{proof}

\section{Application Specifics} \label{applicationsappendix}
In this section, we present in depth how our algorithms is applied to various big data models, and the specifics of the LP-type problems we solve, achieving the bounds in Theorem~\ref{mainresult} and Table~\ref{resultstable}.

\subsection{Application in the Multipass Streaming Model} \label{multipasssection}
The first model we consider is the multipass streaming model. In this model, our input is presented as a stream of points $P$, and our algorithm is allowed to make multiple linear scans of this stream, while maintaining a small space complexity.

Our algorithm is well suited for handling a large stream of data, as it creates a linear sketch of the data points. Therefore, all of our sampling and estimation can easily be done with the aid of our $\ell_0$ samplers and $\ell_0$ estimators. The major consideration we have to make is in how we calculate and store the weights of each snapped point, because we use that to decide which estimators and samplers to use for each point. We clearly cannot store the weight of each snapped point, as that would take up space linear in the number of snapped points. However, one can notice that the weight of a snapped point changes only when it is a violator of a successful iteration, at which point it is multiplied by $N^{\sfrac{1}{s}}$. Thus, it suffices to instead store the solution $f(B)$ for each successful iteration. Now, the weight of a snapped point is simply $N^{\sfrac{v}{s}}$, where $v$ is the number of stored solutions $f(B)$ it violates, which can be calculated on the fly efficiently.

Given this, we can calculate the number of passes needed by our algorithm. First, we see that each iteration of our algorithm requires two passes over the stream. As seen in the last paragraph, getting the weight of each snapped point while creating our sample can be done in one pass. Further, after calculating a solution $f(B)$, checking for violators can also be done in one pass. Finally, we require one initial pass over the data to set our origin point as well as to get the value of $N$.

The space complexity of our algorithm comes from the four quantities it stores, those being the $\ell_0$ estimators, $\ell_0$ samplers, previous solutions $f(B)$, and the current sample $B$. Note that we do not actually need to store the metric \epsnet{} points, as with our construction being a lattice, we can quickly access the closest metric \epsnet{} point to any point $\bmp_i \in P$.

In each iteration we have $t$ $\paren{1 \pm \frac{1}{m^{\sfrac{3}{2}}}}$-approximation $\ell_0$ estimators with $\frac{1}{\polyof{N}}$ error. Each of these estimators has a space usage of $O(m^3 \log^2 N)$. Additionally, we have $2t$ $\paren{1 \pm \frac{1}{4}}$-approximation $\ell_0$ estimators with $\frac{1}{\polyof{N}}$ error. Each of these estimators has a space usage of $O(\log^2 N)$. Finally, we have $t$ $\ell_0$ samplers with  with $\frac{1}{\polyof{N}}$ error, each of which has a space usage of $O(\log(N))$. Given that $t$ is bounded by $O(\nu s)$, and that $m$ is bounded by $O \paren{\nu \lambda N^{\sfrac{1}{s}} \log \paren{\nu \lambda N^{\sfrac{1}{s}}}}$, we have a total storage from all of these of $O \paren{\nu^4 s \lambda^3 N^{\sfrac{3}{s}}} \cdot \polylogof{\nu, \lambda, N}$ words. The $t$ previous solutions can be stored in $O(\nu s) \cdot S_f$ space, where $S_f$ is the number of words needed to store a solution $f(B)$. Finally, a sample of $m$ points can be stored in $O \paren{\nu \lambda N^{\sfrac{1}{s}} \log \paren{\nu \lambda N^{\sfrac{1}{s}}}}$ words. Combining all of these, the total space complexity for our algorithm is 
\[
    O \paren{\nu^4 s \lambda^3 N^{3/s}}\cdot \polylogof{\nu, \lambda, N} + O(\nu s) \cdot S_f + O \paren{\nu \lambda N^{\sfrac{1}{s}} \log \paren{\nu \lambda N^{\sfrac{1}{s}}}}.
\]
We are mostly interested in the multi-pass case where $s = \log N$. First, note that in this case $N^{3/s}$ becomes constant. Further, we have that $\log N = \log \paren{\ceil{\log_{1 + \eps} r_m}} + d \log \paren{1 + \frac{2 \sqrt{d}}{\eps}}$.
The first term diminishes quickly, and since in the high accuracy regime $d < \paren{\sfrac{1}{\eps}}^{0.999}$, we obtain a space complexity of
\[
    O \paren{\nu^4  \lambda^3  d  \log \paren{\frac{1}{\eps}}}\cdot \polylogof{\nu, \lambda, d, \log \paren{\frac{1}{\eps}}} + O \paren{\nu  d  \log \paren{\frac{1}{\eps}}} \cdot S_f + O \paren{\nu \lambda \log \paren{\nu \lambda}}.
\]
When $\nu$ and $\lambda$ are linear in $d$, this space complexity is polynomial in $d$ and polylogarithmic in $\sfrac{1}{\eps}$.

The time complexity of our algorithm depends on three operations we perform in each iteration. The first operation is to create our sample $B$. This can be done in $O(m)$ time per iteration since each point can be sampled from our $\ell_0$ samplers in linear time. The second operation is to calculate the solution for our sample, $f(B)$, which is done once per iteration. The third operation is to check if a snapped point is a violator, i.e.\ to calculate if $f\paren{B \cup \cbrak{\bmq}} > f(B)$. This operation is done for all points each iteration. For a given LP-type problem, we can denote the time bounds for these as $T_B$ and $T_V$. Together, this gives us a total time complexity over $O(\nu s)$ iterations of
\[
    O \paren{\nu s \paren{T_V \cdot n + m + T_B}}.
\]
In order for our algorithm to be efficient in time complexity, we want $T_V$ to be constant time, and $T_B$ to be polynomial in $d$, and logarithmic in $\sfrac{1}{\eps}$.

All together, we achieve the following result for the multipass streaming model.

\begin{theorem}\label{multipassresult}
    For any $s \in [1, d \log \paren{\sfrac{1}{\eps}}]$, there exists a randomized algorithm to compute a $(1 + O(\eps))$-approximation solution to an LP-type problem in the multipass streaming model, that takes $O(\nu s)$ passes over the input, with
    \[
    O \paren{\nu^4 s \lambda^3 N^{3/s}}\cdot \polylogof{\nu, \lambda, N} + O(\nu s) \cdot S_f + O \paren{\nu \lambda N^{\sfrac{1}{s}} \log \paren{\nu \lambda N^{\sfrac{1}{s}}}}
    \]
    space complexity in words and $O \paren{\nu s \paren{T_V \cdot n + m + T_B}}$ time complexity, where $S_f$ is the number of words needed to store a solution $f(B)$.
\end{theorem}

\subsection{Application in the Strict Turnstile Model} \label{turnstilesection}
In the strict turnstile model, there is an underlying vector $\bmv$ where the $i$\textsuperscript{th} element of the vector corresponds to a point $\bmp_i \in P$. This vector is initialized to $\boldsymbol{0}$. The input is then presented a stream of additive updates to the coordinates of $\bmv$, presented as $\bmv \leftarrow \bmv + \bme_i$ or $\bmv \leftarrow \bmv - \bme_i$, where $\bme_i$ is the $i$\textsuperscript{th} standard unit vector. At the end of the stream, we are guaranteed that $\bmv$ is itemwise nonnegative. This stream can also be thought of as operations insert$(\bmp_i)$ and delete$(\bmp_i)$, with the guarantee that at the end of the stream, no point $\bmp_i$ has negative copies. Our algorithm is allowed to make multiple linear scans of this stream, while maintaining a small space complexity.

The application of our algorithm to this model is very similar to the multipass streaming model, because it is based on a linear sketch of the data points. Note that $\ell_0$ samplers and estimators work in the presence of insertions and deletions of points. We again can use the same way of calculating the weights as in the multipass streaming model, meaning that the number of passes over the input stream is the same in each iteration.

The only additional hurdle presented by the strict turnstile model is that we cannot simply pick the first point in the data stream as our origin and center of our metric \epsnet{} and get the maximum distance point in order to define $N$. This is because we need to make sure that the points we use here are not ones that will end up at zero copies at the end of the input stream. Instead, we present the following scheme. We first use an $\ell_0$ sampler in order to sample a distinct non-deleted point $\bmp$, which we can then think of as our origin, and set up our metric \epsnet{} around. We also notice that it suffices to get an $O(1)$-approximation of the largest norm of a (shifted) non-deleted point. Further, note that as each input point is in $\cbrak{-\Delta, -\Delta + 1, \hdots, \Delta}^d$, this norm is upper bounded by $O \paren{\Delta \sqrt{d}}$ and lower bounded by 1. We perform a binary search over the powers of $2$ in this range. Since there are $O(\log (\Delta d))$ powers of $2$ in the range, the binary search will take $O(\log \log (\Delta d))$ iterations. In each iteration of this search, we will use an $\ell_0$ sampler in order to sample a distinct non-deleted point with norm above the current power of $2$. When we find the largest such power of $2$ where there exists a non-deleted point, we will have a $2$-approximation of the largest norm from our origin point $\bmp_1$, and this can set our $N$. Each iteration of the binary search requires one pass over the data, and there is the initial pass to find $\bmp_1$, which gives us a total increase in the pass count of $O(\log \log (\Delta d))$ over the multipass streaming model.

The space and time complexities of our algorithm in the strict turnstile model match the multipass model for each iteration of the algorithm. For the initial creation of the net, we do constant time work to sample one point each iteration, giving us the following result.

\begin{theorem}\label{turnstileresult}
    For any $s \in [1, d \log \paren{\sfrac{1}{\eps}}]$, there exists a randomized algorithm to compute a $(1 + O(\eps))$-approximation solution to an LP-type problem in the multipass streaming model, that takes $O(\nu s + \log \log (\Delta d))$ passes over the input, with
    \[
        O \paren{\nu^4 s \lambda^3 N^{3/s}}\cdot \polylogof{\nu, \lambda, N} + O(\nu s) \cdot S_f + O \paren{\nu \lambda N^{\sfrac{1}{s}} \log \paren{\nu \lambda N^{\sfrac{1}{s}}}}
    \]
    space complexity in words and $O \paren{\nu s \paren{T_V \cdot n + m + T_B} + \log \log (\Delta d) \cdot n}$ time complexity, where $S_f$ is the number of words needed to store a solution $f(B)$.
\end{theorem}

\subsection{Application in the Coordinator Model} \label{coordinatorsection}
In the coordinator model, there are $k$ distributed machines and a separate coordinator machine, which is connected to each machine via a two-way communication channel. The input set $P$ is distributed among the machines, with each machine $i$ getting a part $P_i$, and the goal is to jointly compute the solution $f(P)$ for the LP-type problem. Communication between the machines proceeds in rounds, where each round the coordinator can send a message to each machine, and then each machine can send a message back. The final computation is done by the coordinator. In the coordinator model, our algorithm aims to minimize the number of communication rounds and maximum bits of information sent or received in any round.

Again, notice that we can use the same idea of samplers and estimators at each machine, as well as the idea of on-the-fly weight calculation by storing previous solutions. The major implementation detail now is how to sample $m$ points in the distributed setting and denote success of iterations. We will achieve this by the coordinator calculating what share of the $m$ point sample each machine should be responsible for, and then building up an $m$ point sample from the samples of each machine.

Our algorithm can do both of these in each round with three rounds of communication between the coordinator and the machines. First, at the start of each round, if the last round was deemed successful, the coordinator sends each machine the solution $f(B)$, for each machine to store and use in their weight calculations. The machines send the coordinator the weight of their subset of snapped points, $w(Q_i) = \sum_{\bmq_j \in Q_i} w(\bmq_j)$. Now, the coordinator generates $m$ i.i.d. numbers $x_1, \hdots, x_m$ where each number $i \in k$ has probability of being picked according to the weight share of machine $i$, i.e.\ $\Pr \brak{i \text{ is generated}} = \frac{w(Q_i)}{w(Q)}$ where $w(Q) = \sum_{i \in \brak{k}} w(Q_i)$ is calculated by the coordinator. Starting the second round of communication, the coordinator sends each machine $i$ a number $y_i = \abs{\cbrak{l \mid x_l = i}}$. Each machine then uses its $\ell_0$ samplers to sample $y_i$ points from its subset $Q_i$, according to the weight share of each snapped point $\bmq_j$, i.e.\ $\Pr \brak{\bmq_j \text{ is sampled}} = \frac{w(\bmq_j)}{w(Q_i)}$, and sends this sample to the coordinator. This sampling system ensures that in total, $\Pr \brak{\bmq_j \text{ is sampled}} = \frac{w(\bmq_j)}{w(Q_i)} \cdot \frac{w(Q_i)}{w(Q)} = \frac{w(\bmq_j)}{w(Q)}$, and thus the sample is correctly weighted. The coordinator then combines each machine's sample to create the total sample $B$, and calculates $f(B)$. For the final round of communication, the coordinator sends $f(B)$ to all machines, which they use to calculate their subset's violators. They send back the weight of their violator set, $w(V_i)$. If all of these $w(V_i)$'s are 0, then the algorithm is done. If not, the coordinator can calculate whether the iteration is successful, and continue to the next.

To complete our analysis, notice that the start of the algorithm requires two rounds of communication, one round for any machine to send a point to take as the origin, and a second round for the coordinator to send that point to the machines, and for the machines to send back their furthest distance from it.

To calculate the maximum load over a round of communication, we look at the messages being sent. The start of the algorithm has each message representable as a point, so the maximum load is $k$. Then, in each iteration, we have four different messages to account for. A solution is $S_f$ words in size, so sending it has load $k \cdot S_f$. Each weight being sent is bounded by $w(Q)$, which itself is upper bounded by $N \cdot \paren{N^{\sfrac{1}{s}}}^t$ on the $t$\textsuperscript{th} iteration. Since we have $O(\nu s)$ iterations, we get that the weights are bounded by $N^\nu$, which can fit into $1$ word. Thus, we can bound the load of sending weights at $k$ words. Each number $y_i$ is bounded by $m$, and can thus fit into $1$ word. Finally, the samples being sent by each machine have in total $m$ points, thus the load is $m$ words.

All together, we achieve the following result for the coordinator model.

\begin{theorem}\label{coordinatorresult}
    For any $s \in [1, d \log \paren{\sfrac{1}{\eps}}]$, there exists a randomized algorithm to compute a $(1 + O(\eps))$-approximation solution to an LP-type problem in the coordinator model, that takes $O(\nu s)$ rounds of communication, with 
    \[
    O \paren{\nu \lambda N^{\sfrac{1}{s}} \log \paren{\nu \lambda N^{\sfrac{1}{s}}} + k \cdot S_f}
    \]
    load in words, where $S_f$ is the number of words needed to store a solution $f(B)$. The local computation time of the coordinator is $O \paren{\nu s \paren{m + T_B + k}}$, and the local computation time of each machine $i$ is $O \paren{\nu s \paren{n_i T_V}}$ where $n_i = \abs{P_i}$.
\end{theorem}

\subsection{Application in the Parallel Computation Model} \label{parallelsection}
In the parallel computation model, there are $k$ distributed machines, which are each connected via two-way communication channels. The input set $P$ is distributed among the machines, with each machine $i$ getting a part $P_i$, and the goal is to jointly compute the solution $f(P)$ for the LP-type problem. Communication between the machines proceeds in rounds, where each round the machines can communicate with each other in messages. In the coordinator model, our algorithm aims to minimize the number of communication rounds and maximum bits of information sent or received in any round.

Our procedure for running our algorithm in the parallel computation model is simply to run it as our algorithm for the coordinator model, assigning one of the machines to act as the coordinator. Thus, we can achieve the same bounds on the rounds of communication and maximum load. For the computation times, one machine will have a computation time on the order of the coordinator time added to the machine time, and other machines will have computation times same as in the coordinator model.

\subsection{Solving the MEB Problem}\label{MEBapplication}
One important application of our algorithm is in solving the MEB problem. Given $n$ points ${\bmp_i \in \cbrak{-\Delta, -\Delta + 1, \hdots, \Delta}^d}$, the objective of the problem is to find a $d$-dimensional ball \linebreak${(\bmc, r) \in \bbR^d \times \bbR}$ with center $c$ and radius $r$ that encloses all of the input points, i.e.\ where for all points $\bmp_i$, $d(\bmc, \bmp_i) \leq r$, and which has the smallest such radius $r$.

The MEB problem is an LP-type problem where $S$ is the set of $d$-dimensional points, and $f(\cdot)$ is the function that maps from a set of points to their MEB. The combinatorial dimension of the MEB problem is $\nu = d + 1$, as any MEB in $d$ dimensions defined by $d+2$ points $P$ can be defined by any $d+1$ point subset of $P$. However, there are MEB instances defined by $d+1$ points that cannot be defined by a $d$ point subset, e.g., $3$ vertices of an equilateral triangle in $2$ dimensions. Further, the VC dimension of the MEB set system is $\lambda = d + 1$~\cite{VC15}.

Now, we show that the MEB problem satisfies the properties~\ref{P1},~\ref{P2}, and~\ref{P3}. For~\ref{P1}, each point can be associated with the set of all balls that enclose it. For~\ref{P2}, it can be seen that for a set of points, the intersection of sets of balls that enclose them is the set of balls that enclose the whole set. Thus, it is natural that the MEB is the minimal such ball. Finally, for~\ref{P3}, we show that for an MEB $(\bmc,r)$ of a set of snapped points $Q$, the ball $(\bmc, (1 + 4\eps)r)$ is a $(1 + O(\eps))$-approximation MEB of the set of original points $P$. We show this using two lemmas.

\begin{restatable}{lemma}{containsrestatable} \label{containslemma}
    All points $\bmp_i \in P$ are inside the ball $(\bmc, (1 + 4 \eps) r)$, i.e., $d(\bmc, \bmp_i) \leq (1 + 4 \eps) r$.
\end{restatable}

\begin{restatable}{lemma}{nottoobigrestatable} \label{nottoobiglemma}
    $(\bmc, r)$ is no larger than $(1 + O(\eps))$ times the size of the MEB of the original points $(\bmc^*, r^*)$, i.e., $r \leq (1 + O(\eps))r^*$.
\end{restatable}

Together, these lemmas show that the ball $(\bmc, (1 + 4 \eps)r)$ is an enclosing ball for $P$, and is within $1 + O(\eps)$ of the MEB for $P$, which together show~\ref{P3}.

Thus, our algorithm can be utilized to solve the $(1 + 4 \eps)$-approximation MEB problem. Notice that a solution to the MEB problem is a ball $(\bmc, r)$ and can thus be stored in $2$ words of space. Further, violator checks can easily be done on a stored solution by checking for a snapped point $\bmq$, whether $d(\bmc, \bmq) \leq r$, giving us $T_V = O(1)$ for the MEB problem. The MEB of a sample of $m$ points can be found in $O \paren{ \paren{m + d}^3}$~\cite{YT89}, which gives us $T_B$ for the MEB problem. Therefore, we can achieve the following result for the MEB problem, using the results from Theorems~\ref{multipassresult},~\ref{turnstileresult}, and~\ref{coordinatorresult}.

\begingroup \abovedisplayskip=2pt \belowdisplayskip=0pt
\begin{theorem}\label{MEBresult}
    For any $s \in [1, d \log \paren{\sfrac{1}{\eps}}]$, there exists a randomized algorithm to compute a $(1 + 4\eps)$-approximation solution to the MEB problem with high probability in the following models, where $N = \ceil{\log_{1 + \eps} r_m} \paren{1 + \frac{2\sqrt{d}}{\eps}}^d$.
    \begin{itemize}
        \item Multipass Streaming: An $O\paren{d s}$ pass algorithm with 
        \[
            O \paren{d^7 s N^{3/s}}\cdot \polylogof{d, N} + O \paren{ds + d^2 N^{\sfrac{1}{s}} \log \paren{d N^{\sfrac{1}{s}}}}
        \]
        space complexity in words and $O \paren{d s \paren{n + \paren{m + d}^3}}$ time complexity.
        
        \item Strict Turnstile: An $O\paren{d s + \log \log \paren{\Delta d}}$ pass algorithm with 
        \[
            O \paren{d^7 s N^{3/s}}\cdot \polylogof{d, N} + O \paren{ds + d^2 N^{\sfrac{1}{s}} \log \paren{d N^{\sfrac{1}{s}}}}
        \]
        space complexity in words and $O \paren{d s \paren{n + \paren{m + d}^3}}$ time complexity.
        
        \item Coordinator / Parallel Computation: An algorithm with $O\paren{d s}$ rounds of communication and
        \[
            O \paren{d^2 N^{\sfrac{1}{s}} \log \paren{d N^{\sfrac{1}{s}}} + k}
        \]
        load in words. The local computation time of the coordinator is $O \paren{d s \paren{\paren{m + d}^3 + k}}$, and the local computation time of each machine $i$ is $O \paren{d s n_i}$ where $n_i = \abs{P_i}$.
    \end{itemize}
\end{theorem}
\endgroup

\subsection{Solving the Linear SVM Problem}\label{svmapplication}
A seperate application of our algorithm is as a Monte Carlo application in solving the Linear SVM Problem. In the problem, we are given $n$ points $\bmp_i = \paren{\bmx_i, y_i}$ where \linebreak${\bmx_i \in \cbrak{-1, -1 + \frac{1}{\Delta}, \hdots, 1}^d}$ and $y_i \in \cbrak{-1, +1}$. We are also given some $\gamma > 0$ for which the points are either inseparable or $\gamma$-separable. The objective of the problem is to determine either that the points are inseparable, or to compute a $d$-dimensional hyperplane $(\bmu,b)$ where $b$ is the bias term, which separates the points with the largest margin, i.e.\ to solve the quadratic optimization problem
\begin{equation} \label{originalSVM}
    \min_{\bmu \in \bbR^d} \norm{\bmu}^2 \quad \text{subject to} \quad y_i \paren{\bmu^T \bmx_i - b} \geq 1 \quad \text{for all} \quad i\in \brak{n}.
\end{equation}
The linear SVM problem is an LP-type problem where $S$ is the set of tuples of $d$-dimensional points and $\pm 1$ $y$ values, and $f(\cdot)$ is the function that maps from a set of points to their optimal separating hyperplane, or to the value infeasible. The combinatorial dimension of the linear SVM problem is $\nu = d + 1$ for a separable set, and $\nu = d + 2$ for an inseparable set. For a separable set, this can be seen as any basis is in the form of a hyperplane with $y = \pm 1$ and a point with $y = \mp 1$. For an inseparable set, another point with $y = \mp 1$ on the opposite side of the hyperplane is needed. Thus, $\nu \leq d + 2$. Further, the VC dimension of the linear SVM set system is $\lambda = d + 1$~\cite{VC15}.

For the linear SVM problem, there is an additional hurdle to overcome. That is, if the points are too close together, running our algorithm with a metric \epsnet{} might not retain separability. Therefore, we instead create a metric $\frac{\eps \gamma}{2}$-net centered at the origin, where $\eps < \frac{1}{2}$. This means that we have snapped points $\bmq_i = (\bmz_i, y_i)$ where $\bmz_j = \bmx_i + \bme_i$ for some metric net point $\bme_i$ with $\norm{\bme_i} \leq \frac{\eps \gamma}{2}$. Our algorithm then finds the hyperplane that solves the quadratic optimization problem
\begin{equation} \label{snappedSVM}
    \min_{\bmu \in \bbR^d} \norm{\bmu}^2 \quad \text{subject to} \quad y_i \paren{\bmu^T \bmz_i - b} \geq 1 \quad \text{for all} \quad i\in \brak{N}.
\end{equation}
The added benefit of this is that we do not need to use the additional search in the turnstile model, and can use the same metric net. We also run our algorithm as a Monte Carlo algorithm, where if we do not find a solution in the first $O(d s)$ iterations, we return that the point set is not separable. This gives us an algorithm with a one-sided error. That is, if a point set is inseparable, we will always output as such. If it is separable however, we will output a $(1 + O(\eps))$-approximation for the optimal separating hyperplane with high probability.

Given this new way of running our algorithm, we now show that the linear SVM problem satisfies the properties~\ref{P1},~\ref{P2}, and~\ref{P3}. As a constrained optimization problem, it satisfies~\ref{P1} and~\ref{P2} naturally. Each point is associated with the set of hyperplanes $(\bmu, b)$ that satisfy its constraint, and for a set of constraints, the solution minimizes $\norm{\bmu}^2$ over hyperplanes that satisfy all constraints. Finally, for~\ref{P3}, we show that for the optimal hyperplane $(\bmu,b)$ separating a set of snapped points $Q$, the hyperplane $((1 + 2\eps)\bmu, (1+2\eps)b)$ is a $(1 + O(\eps))$-approximation of the optimal hyperplane separating the set of original points $P$. We show this using two lemmas.

\begin{restatable}{lemma}{SVMcontainsrestatable} \label{SVMcontainslemma}
    The hyperplane $((1 + 2\eps)\bmu, (1+2\eps)b)$ separates all original points $\bmp_i \in P$, i.e., $y_i \paren{\paren{1 + 2\eps}\bmu^T \bmx_i - \paren{1 + 2\eps}b} \geq 1 \text{ for all } i\in \brak{n}$.
\end{restatable}

\begin{restatable}{lemma}{SVMnottoobigrestatable} \label{SVMnottoobiglemma}
    The objective of Problem~\ref{snappedSVM} is no larger than $(1 + O(\eps))$ times the objective of Problem~\ref{originalSVM}, i.e., $\norm{\bmu}^2 \leq (1 + O(\eps))\norm{\bmu^*}^2$.
\end{restatable}

Together, these lemmas show that the hyperplane $((1 + 2\eps)\bmu, (1+2\eps)b)$ is a separating hyperplane for $P$, and is within $1 + O(\eps)$ of the objective function for Problem~\ref{originalSVM}, which together show~\ref{P3}.

Thus, our algorithm can be utilized to solve the $(1 + 2 \eps)$-approximation linear SVM problem. Notice that a solution to the linear SVM problem is a hyperplane $(\bmu, b)$ and can thus be stored in $2$ words of space. Further, violator checks can easily be done on a stored solution by checking for a snapped point $\bmq = (\bmz, y)$, whether $y \paren{\bmu^T \bmz - b} \geq 1$, giving us $T_V = O(1)$ for the linear SVM problem. The optimal separating hyperplane of a sample of $m$ points can be found in $O \paren{ \paren{m + d}^3}$~\cite{YT89}, which gives us $T_B$ for the linear SVM problem. Therefore, we can achieve the following result for the linear SVM problem, using the results from Theorems~\ref{multipassresult},~\ref{turnstileresult}, and~\ref{coordinatorresult}.

\begingroup \abovedisplayskip=2pt \belowdisplayskip=0pt
\begin{theorem}\label{SVMresult}
    For any $s \in [1, d \log \paren{\sfrac{1}{\eps}}]$, there exists a randomized algorithm to compute a $(1 + 2\eps)$-approximation solution to the linear SVM problem with high probability in the following models, where $N = \paren{1 + \frac{4\sqrt{d}}{\eps \gamma}}^d$.
    \begin{itemize}
        \item Multipass Streaming: An $O\paren{d s}$ pass algorithm with
        \[
            O \paren{d^7 s N^{3/s}}\cdot \polylogof{d, N} + O \paren{ds + d^2 N^{\sfrac{1}{s}} \log \paren{d N^{\sfrac{1}{s}}}}
        \]
        space complexity in words and $O \paren{d s \paren{n + \paren{m + d}^3}}$ time complexity.
        
        \item Strict Turnstile: An $O\paren{d s}$ pass algorithm with
        \[
            O \paren{d^7 s N^{3/s}}\cdot \polylogof{d, N} + O \paren{ds + d^2 N^{\sfrac{1}{s}} \log \paren{d N^{\sfrac{1}{s}}}}
        \]
        space complexity in words and $O \paren{d s \paren{n + \paren{m + d}^3}}$ time complexity.
        
        \item Coordinator / Parallel Computation: An algorithm with $O\paren{d s}$ rounds of communication and
        \[
            O \paren{d^2 N^{\sfrac{1}{s}} \log \paren{d N^{\sfrac{1}{s}}} + k}
        \]
        load in words. The local computation time of the coordinator is $O \paren{d s \paren{\paren{m + d}^3 + k}}$, and the local computation time of each machine $i$ is $O \paren{d s n_i}$ where $n_i = \abs{P_i}$.
    \end{itemize}
\end{theorem}
\endgroup

\subsection{Solving Bounded Linear Programming Problems up to Additive Epsilon Error}
So far, we have dealt with applications with multiplicative error, i.e.\ where our solution is within $(1 + O(\eps))$ times the optimal solution. We can also use our algorithm to solve problems up to additive error. One application for this usage is in certain linear programming applications we will explore. Linear programs in general are optimization problems of the form
\begin{equation} \label{originalLP}
    \max_{x \in \bbR^d} \bmc^T \bmx \quad \text{subject to} \quad \bma_i^T \bmx \leq b_i \quad \text{for all} \quad i\in \brak{n}
\end{equation}
where the input is the objective vector $\bmc$ as well as $n$ input constraints $\bmp_i = \paren{\bma_i, b_i}$.

We will work with a certain class of linear programs, where the input points, as well as the solution vector $\bmx$ is bounded, i.e., for all constraints $i$, $\norm{\bma_i}, \abs{b_i} \in O(1)$, as well as $\norm{\bmc}, \norm{\bmx} \in O(1)$. Given this, we can run our algorithm by snapping each constraint point $\bmp_i$ to a metric $O(\eps)$-net to get $\bma_i' = \bma_i + \bme_i$ where $\norm{\bme_i} \leq O(\eps)$ and $b_i' = b_i + f_i$ where $\abs{f_i} \leq O(\eps)$. As a note, this construction means that we require no additional passes in the turnstile model to create our net, similar to the linear SVM problem. This snapping then gives us the LP
\begin{equation} \label{snappedLP}
    \max_{\bmx \in \bbR^d} \bmc^T \bmx \quad \text{subject to} \quad \bma_i'^T \bmx \leq b_i' \quad \text{for all} \quad i\in \brak{N}.
\end{equation}
The optimal solution to LP~\eqref{snappedLP}, $\bmx$, then gives an additive $\eps$-approximation solution to LP~\eqref{originalLP} as well.

Linear programs are the eponymous LP-type program, where $S$ is the set of constraints of the LP and $f(\cdot)$ is the function that maps the set of constraints to their optimal solution, or to the value infeasible. An important note is that there might be multiple optimal solutions for a set of constraints, in which case any tie-breaking system may be used (a common one is to take the lexicographically smallest optimal point). The combinatorial dimension of linear programming is $\nu = d$, as $d$ is the number of variables of the LP, and the VC dimension is $\lambda = d + 1$, since each constraint for an LP induces a feasible half-space~\cite{MSW96, VC15}.

We now show that linear programs satisfy the properties~\ref{P1},~\ref{P2}, and~\ref{P3}. As a constrained optimization problems, they satisfy~\ref{P1} and~\ref{P2} naturally. Each constraint is associated with the halfspace that satisfies it, and for a set of constraints, the objective function is maximized over points that satisfy all constraints. Finally, we show a modified~\ref{P3} as we are now deal with additive approximations. We show that for the optimal solution $\bmc^T\bmx$ satisfying a set of snapped constraints $Q$, the solution $\bmx$ gives an additive $\eps$-approximation to the optimal solution $\bmc^T \bmx^*$ approximately satisfying the original constraints $P$. We show this using two lemmas.

\begin{restatable}{lemma}{LPcontainsrestatable} \label{LPcontainslemma}
    The solution $\bmx$ approximately satisfies all original constraint points $\bmp_i \in P$, i.e., $\bma_i^T \bmx \leq b_i + O(\eps)$ for all $i\in \brak{n}$.
\end{restatable}

\begin{restatable}{lemma}{LPnottoobigrestatable} \label{LPnottoobiglemma}
    The objective of LP~\eqref{snappedLP} is at most $O(\eps)$ smaller than the objective of LP~\eqref{originalLP}, i.e., $\bmc^T \bmx \geq \bmc^T \bmx^* - O(\eps)$.
\end{restatable}

Together, these lemmas show that the solution $\bmx$ satisfies the constraints $P$, and is within additive $O(\eps)$ of the objective function for LP~\eqref{originalLP}, which together show~\ref{P3}.

Thus, our algorithm can be utilized to solve additive $\eps$-approximation LP problems. Notice that a solution to an LP is a point $\bmx$ and can thus be stored in $1$ word of space. Further, violator checks can easily be done on a stored solution by checking whether it satisfies a given constraint, giving us $T_V = O(1)$ for LP problems. The optimal solution to an LP with $d$ variables and $m$ constraints can be found in $\Tilde{O}\paren{md + d^{2.5}}$~\cite{BLLSSSW21}, which gives us $T_B$ for LP problems. Therefore, we can achieve the following result for bounded LP problems, using the results from Theorems~\ref{multipassresult},~\ref{turnstileresult}, and~\ref{coordinatorresult}.

\begingroup \abovedisplayskip=2pt \belowdisplayskip=0pt
\begin{theorem}\label{LPresult}
    For any $s \in [1, d \log \paren{\sfrac{1}{\eps}}]$, there exists a randomized algorithm to compute an additive $\eps$-approximation solution to bounded LP problems with high probability in the following models, where $N \in O\paren{\frac{\sqrt{d}}{\eps}}^d$.
    \begin{itemize}
        \item Multipass Streaming: An $O\paren{d s}$ pass algorithm with
        \[
            O \paren{d^7 s N^{3/s}}\cdot \polylogof{d, N} + O \paren{ds + d^2 N^{\sfrac{1}{s}} \log \paren{d N^{\sfrac{1}{s}}}}
        \]
        space complexity in words and $\Tilde{O} \paren{d s \paren{n + md + d^{2.5}}}$ time complexity.
        
        \item Strict Turnstile: An $O\paren{d s}$ pass algorithm with
        \[
            O \paren{d^7 s N^{3/s}}\cdot \polylogof{d, N} + O \paren{ds + d^2 N^{\sfrac{1}{s}} \log \paren{d N^{\sfrac{1}{s}}}}
        \]
        space complexity in words and $\Tilde{O} \paren{d s \paren{n + md + d^{2.5}}}$ time complexity.
        
        \item Coordinator / Parallel Computation: An algorithm with $O\paren{d s}$ rounds of communication and
        \[
            O \paren{d^2 N^{\sfrac{1}{s}} \log \paren{d N^{\sfrac{1}{s}}} + k}
        \]
        load in words. The local computation time of the coordinator is $\Tilde{O} \paren{d s \paren{md + d^{2.5} + k}}$, and the local computation time of each machine $i$ is $O \paren{d s n_i}$ where $n_i = \abs{P_i}$.
    \end{itemize}
\end{theorem}
\endgroup

\subsubsection{Solving the Linear Classification Problem} \label{linearclassificationsection}
One example of a bounded LP where our algorithm can be used to achieve a useful result is the linear classification problem. In the problem, we are given $n$ labeled examples $\bmp_i$, comprising of a bounded $d$-dimensional feature vector $\bmx_i \in \cbrak{-1, -1 + \frac{1}{\Delta}, \hdots, 1}^d$ and a corresponding label $y_i \in \cbrak{-1, +1}$. The goal of the problem is to find a separating hyperplane $\bmu$, which can be thought of as a normal vector $\bmu \in \bbR^d$ where $\norm{\bmu} = 1$, such that $y_i \paren{\bmx_i^T \bmu} \geq 0$ for all points $\bmp_i$. We will consider the related approximate optimization problem, where we assume that there is an optimal classifier $\bmu^*$ such that $y_i \paren{\bmx_i^T \bmu^*} \geq \eps$ for all points $\bmp_i$. We thus want to find a classifier that is within additive $\eps$ of the separation of this optimal classifier.

This problem can be written as a bounded LP as such. First, notice that we have two types of constraints, $\bmx_i^T \bmu \leq 0$ for any point $\bmp_i$ with $y_i = -1$, and $\bmx_i^T u \geq 0$ for any point $\bmp_i$ with $y_i = +1$. For simplicity, we will actually consider the negation of any such example, so we will actually consider points $\bmp_i' = \paren{\bmx_i', y_i'}$ where $\bmx_i' = -\bmx_i$ and $y_i' = -y_i$ if $y_i = +1$, and $\bmx_i' = \bmx_i$ and $y_i' = y_i$ otherwise. This allows us to only have constraints of the type
\[
    \quad \bmx_i'^T \bmu \leq 0 \quad \text{for all} \quad i\in \brak{n}.
\]
Our objective is to maximize the separation of $\bmu$, which is $min_{i \in \brak{n}} \bmx_i'^T \bmu$. While this objective isn't linear as is, it can be made linear by utilizing an additional variable representing the separation, and $n$ additional constraints.

Thus, we can use our general framework for bounded LP problems in order to solve this problem with separation that is within additive $\eps$ of the optimal hyperplane's separation. This gives us a very important benefit. Namely, since the optimal hyperplane's separations is $\eps$, our hyperplane will in fact be a separating hyperplane for the original points, meaning that our solution fully satisfies the original constraints $P$. Thus, we are able to solve the linear classification problem with an additive $\eps$-approximation of the largest separation, in the bounds given in Theorem~\ref{LPresult}.

\subsection{Solving Bounded Semidefinite Programming Problems up to Additive Epsilon Error} \label{sdpsubsection}
Another additive $\eps$-approximation application of our algorithm is in solving bounded semidefinite programming problems. These problems are optimization problems of the form
\begin{equation} \label{originalSDP}
    \max_{X \in \bbR^{d \times d}} \Frobinner{C}{X} \quad \text{subject to} \quad \genfrac{}{}{0pt}{}{\Frobinner{A^{(i)}}{X} \leq b^{(i)} \quad \text{for all} \quad i\in \brak{n}}{X \succeq 0}
\end{equation}
where $\Frobinner{\cdot}{\cdot}$ is the Frobenius inner product, i.e.\ $\Frobinner{A}{B} = \sum_{i,j} A_{ij} B_{ij}$, and $X \succeq 0$ denotes $X$ as a positive semidefinite matrix. The input will be the objective matrix $C$ as well as $n$ input constraints $\bmp_i = \paren{A^{(i)}, b^{(i)}}$.

We will work with a certain bounded class of SDP problems. Namely, we will have the following boundedness assumptions.
\begin{itemize}
    \item $X$ has unit trace, i.e.\ $\Tr(X) = 1$,
    \item $\norm{C}_F \leq 1$ where $\norm{\cdot}_F$ is the Frobenius norm.
\end{itemize}
Further, for all constraints $i$
\begin{itemize}
    \item $\norm{\Ai}_2 \leq 1$ where $\norm{\cdot}_2$ is the spectral norm,
    \item $\norm{\Ai}_F \leq F$,
    \item The number of nonzero entries of $\Ai$ is bounded by $S$,
    \item $\abs{\bi} \leq 1$.
\end{itemize}

\begin{definition}
    These bounds and later calculations use several definitions of norms. For an $n \times m$ matrix $A$, we use the following definitions of norm.
    \begin{itemize}
        \item Entry-wise $1$-norm: $\norm{A} = \sum_{i \in \brak{n}, j \in \brak{m}} A_{ij}$.
        \item Spectral norm: $\norm{A}_2 = \max_{\bmv \in \bbR^m, \norm{\bmv} \neq 0} \frac{\norm{A \bmv}}{\norm{\bmv}}$.
        \item Frobenius norm: $\norm{A}_F = \sqrt{\sum_{i \in \brak{n}, j \in \brak{m}} A_{ij}^2}$.
        \item Schatten $1$-norm: $\norm{A}_* = \sum_{i = 1}^{\min(n, m)} \sigma_i(A)$, where $\sigma_i(A)$ are the singular values of $A$.
    \end{itemize}
\end{definition}

Our algorithm solves these problems by solving them as bounded LP problems with the solution vector $\bmx \in \bbR^{d^2}$ where $\bmx = \Vectorize(X)$ is the vectorization of $X$. This gives us the natural objective function $\max_{\bmx \in \bbR^{d^2}} \bmc^T \bmx$ where $\bmc = \Vectorize(C)$. For the constraints of these LPs, we consider the two different type of SDP constraints separately. First, we have constraints of the type $\Frobinner{\Ai}{X} \leq \bi$, which we can represent as $\bma^{(i) T} \bmx \leq \bi$ where $\bma^{(i)} = \Vectorize(\Ai)$. We also have the constraint that $X$ must be positive semidefinite, i.e.\ $X \succeq 0$. This is equivalent to the constraints $\bmy^T X \bmy \geq 0$ for all vectors $\bmy \in \bbR^d$ where $\norm{\bmy} = 1$. We can represent each of these as constraints $(\bmy \otimes \bmy)^T x$. The problem however is that we would need an infinite amount of these constraints, as there are infinite $\bmy$ vectors to consider.

Thus, we again utilize metric $\eps$-nets. For each $\Ai$, we note that a naive lattice where each coordinate $\Ai_{jk}$ is snapped to the nearest multiple of $\frac{\eps}{d}$ would suffice, however we can achieve a tighter bound because of the assumption that $\Ai$ has at most $S$ nonzero entries. We can therefore create a net for each of the $\binom{d^2}{S}$ possible combinations, where each net size is exponential in $S$ rather than $d^2$. Now, for any $\Ai$, its nonzero coordinates must be fully captured by at least one of these nets, and so we can deterministically choose one and snap it to a point in that net. We can further snap to the nearest multiple of $\frac{\eps}{\min(d, S)}$ in order to decrease the size of each net. Finally, notice that for each coordinate of $\Ai$, $\abs{\Ai_{jk}} \leq \norm{\Ai}_2 \leq 1$. Thus each net is of size $\paren{\frac{2 \min(d, S)}{\eps} + 1}^ S$, and therefore we have a total size of $\binom{d^2}{S} \cdot O\paren{\paren{\frac{\min(d, S)}{\eps}}^S}$ for the nets. For each $b_i$, we can snap it to the nearest multiple of $\eps$, giving a net of size $O\paren{\frac{1}{\eps}}$. In total, this means that in our snapped LP, we have $\binom{d^2}{S} \cdot O\paren{\paren{\frac{\min(d, S)}{\eps}}^S \frac{1}{\eps}}$ constraints of the form $\Frobinner{A'^{(i)}}{X} \leq b'^{(i)}$.

For the positive semidefinite constraints, we create a lattice where each coordinate is a multiple of $\frac{\eps}{d\sqrt{d}}$. This gives us a metric \epsnet{} of size $O \paren{ \paren{\frac{d\sqrt{d}}{\eps}}^d}$, and so we can reduce the infinite constraint space into that many constraints of the form $\bmz^T X \bmz \geq 0$ where $\bmz \in \bbR^d$ are points on the metric \epsnet{}. We can call the original (infinitely numerous) constraints $Y$, and the new constraints $Z$.

Because our algorithm now solves an LP, our combinatorial dimension is $\nu = d^2$, and our VC dimension is $\lambda = d^2 + 1$. Further, the problems satisfy the properties~\ref{P1} and~\ref{P2}. Finally, because of the different formulation of our nets, we have to show~\ref{P3}, again. We show that for the optimal solution $\Frobinner{C}{X}$ satisfying a set of snapped constraints $Q$ and $Z$, the solution $X + \frac{3 \eps}{d} I$, where $I$ is the $d \times d$ identity matrix, gives an additive $\eps$-approximation to the optimal solution $\Frobinner{C}{X^*}$ approximately satisfying the original constraints $P$ and $Y$. We show this using three lemmas.

\begin{restatable}{lemma}{SDPcontainsrestatable} \label{SDPcontainslemma}
    The solution $X + \frac{3 \eps}{d} I$ approximately satisfies all original constraint points $\bmp_i \in P$, i.e., $\Frobinner{\Ai}{X + \frac{3 \eps}{d} I} \leq \bi + O(\eps)$ for all $i \in \brak{n}$.
\end{restatable}

\begin{restatable}{lemma}{SDPpsdrestatable} \label{SDPpsdlemma}
    The solution $X + \frac{3 \eps}{d} I$ satisfies all constraints $\bmy \in Y$, i.e., $\bmy^T (X + \frac{3 \eps}{d} I)\bmy \geq 0$ for all $\bmy \in \bbR^d$ such that $\norm{\bmy} = 1$.
\end{restatable}

\begin{restatable}{lemma}{SDPnottoobigrestatable} \label{SDPnottoobiglemma}
    The objective of the snapped LP is at most $O(\eps)$ smaller than the objective of SDP~\eqref{originalSDP}, i.e., $\Frobinner{C}{X + \frac{3 \eps}{d} I} \geq \Frobinner{C}{X^*} - O(\eps)$.
\end{restatable}

Together, these lemmas show that the solution $X + \frac{3 \eps}{d} I$ satisfies the constraints $P$ and $Y$, and is within additive $O(\eps)$ of the objective function for SDP~\eqref{originalSDP}, which together show~\ref{P3}.

Thus, our algorithm can be utilized to solve additive $\eps$-approximation SDP problems. Notice that a solution to an SDP is a point $(X + \frac{3 \eps}{d} I)$ and can thus be stored in $1$ word of space. Further, violator checks can easily be done on a stored solution by checking whether it satisfies a given constraint, giving us $T_V = O(1)$ for SDP problems. For the bounds of the problem, notice that we convert an SDP problem into an LP problem on $d^2$ dimensions. Thus, we can use the $T_B$ for LPs with $d^2$ variables and $m$ constraints. Therefore, we can achieve the following result for bounded SDP problems, using the results from Theorems~\ref{multipassresult},~\ref{turnstileresult}, and~\ref{coordinatorresult}.

\begingroup \abovedisplayskip=2pt \belowdisplayskip=0pt
\begin{theorem}\label{SDPresult}
    For any $s \in [1, d^2 \log \paren{\sfrac{1}{\eps}}]$, there exists a randomized algorithm to compute an additive $\eps$-approximation solution to bounded SDP problems with high probability in the following models, where $N \in \binom{d^2}{S} \cdot O\paren{\paren{\frac{\min(d, S)}{\eps}}^S \frac{1}{\eps}} + O \paren{ \paren{\frac{d\sqrt{d}}{\eps}}^d}$.
    \begin{itemize}
        \item Multipass Streaming: An $O\paren{d^2 s}$ pass algorithm with
        \[
            O\paren{d^{14} s N^{\sfrac{3}{s}}} \cdot \polylogof{d, N} + O \paren{d^2 s + d^4 N^{\sfrac{1}{s}} \log \paren{d N^{\sfrac{1}{s}}}}
        \]
        space complexity in words and $\Tilde{O} \paren{d^2 s \paren{n + md^2 + d^{5}}}$ time complexity.
        
        \item Strict Turnstile: An $O\paren{d^2 s}$ pass algorithm with
        \[
            O\paren{d^{14} s N^{\sfrac{3}{s}}} \cdot \polylogof{d, N} + O \paren{d^2 s + d^4 N^{\sfrac{1}{s}} \log \paren{d N^{\sfrac{1}{s}}}}
        \]
        space complexity and $\Tilde{O} \paren{d^2 s \paren{n + md^2 + d^{5}}}$ time complexity.
        
        \item Coordinator/Parallel Computation: An algorithm with $O\paren{d^2 s}$ rounds of communication and
        \[
            O \paren{d^4 N^{\sfrac{1}{s}} \log \paren{d N^{\sfrac{1}{s}}} + k}
        \]
        load in words. The local computation time of the coordinator is $O \paren{d^2 s \paren{md^2 + d^{5} + k}}$, and the local computation time of each machine $i$ is $O \paren{d^2 s n_i}$ where $n_i = \abs{P_i}$.
    \end{itemize}
\end{theorem}
\endgroup

\subsubsection{Solving SDP Saddle Point Problems in the Unit Simplex} \label{saddlepointsection}
One example of a bounded SDP where our algorithm can be used to achieve a useful result is the saddle point problem
\begin{equation} \label{saddlepointequation}
    \max_X \min_{\bmp \in \Delta_n} \sum_{i \in \brak{n}} p_i \paren{\Frobinner{\Ai}{X} - \bi}
\end{equation}
where for all $i \in \brak{n}$, $\Ai \in \bbR^{d \times d}$ are symmetric and $\bi \in \bbR$. $X \in \bbR^{d \times d}$ is a positive semidefinite matrix of trace $1$, and $\Delta_n = \cbrak{\bmx \in \bbR^n \mid \bmx \geq 0, \norm{\bmx}_1 = 1}$ is the $n - 1$ dimensional unit simplex. In the case where the optimal solution to Equation~\ref{saddlepointequation} is nonnegative, then solving it up to additive $\eps$ error is equivalent to finding an $X$ that satisfies all constraints $\Frobinner{\Ai}{X} \geq \bi$ up to additive $\eps$ error. The optimization version of this is maximizing a margin $\sigma$ where $\Frobinner{\Ai}{X} \geq \bi + \sigma$.

Notice that the constraints of the problem directly fit into our framework for bounded SDP problems, and the margin $\sigma$ can be represented as an aditional variable in our LP formulation. The maximization objective $\max \sigma$ now requires a $1$-hot vector $c$, which fits into our boundedness constraint.

Thus, we can use our general framework for bounded SDP problems in order to solve this problem up to an additive $\eps$-approximation, in the bounds given in Theorem~\ref{SDPresult}.

\section{Lower Bound Proofs} \label{lbappendix}

\subsection{Lower Bound for the MEB Problem}\label{meblowersubsection}
In this section, we provide the following lower bound for the space complexity of the $1$-pass streaming MEB problem via the communication complexity of the $1$-round MEB problem.

\ccmebrestatable*

\begin{proof}
    To prove this, we prove a lower bound on any $1$-round communication protocol which yields a $(1 + \eps)$-approximation for the MEB problem.
    
    The proof proceeds by reduction from $\mathsf{Ind}_n$ to MEB. Given an instance of the indexing problem, Alice uses a scheme in order to transform her bitstring $b$ into a set of at most $n$ points $P \in \bbR^d$, one point $\bmp_j$ for each $b_j = 1$.
    Separately, Bob uses a scheme in order to transform his index $i$ into a single point $\bmq \in \bbR^d$. 
    We then show that if $b_i = 1$, then the MEB of the at most $n+1$ points $P \cup \cbrak{\bmq}$ obtained by Alice and Bob has radius $2$, while if $b_i = 0$ then the MEB has radius at most $2 - \Omega(\varepsilon)$ where $\varepsilon$ satisfies $n = \left(\sfrac{1}{\varepsilon}\right)^{\lfloor d / 4 \rfloor}$. 
    This therefore gives a $\paren{\sfrac{1}{\eps}}^{\Omega(d)}$ lower bound on the communication complexity of any $1$-round communication protocol for a $(1 + \eps)$-approximation, which gives the same lower bound on the space complexity of any $1$-pass streaming algorithm for the $(1 + \eps)$-approximation MEB problem.
    We now provide a full proof, first for $d=2$, and then extended to higher dimensions.

    We proceed with the proof when $d = 2$.

    Let $\eps = \paren{\frac{1}{n}}^2$ and Alice's points $\bmp^{(1)}, \bmp^{(2)}, \hdots, \bmp^{(n)} \in \bbR^2$ be equally spaced points on the unit circle, where
    \[ \bmpj = 
    \left(
    \cos\left(2\pi j \sqrt{\varepsilon}\right),\, 
    \sin\left(2\pi j \sqrt{\varepsilon}\right)
    \right). \]
    For each $j \in [n]$, if $b_j = 1$, then Alice adds $\bmpj$ to her portion of the input stream to the MEB problem. Bob adds to his portion of the input stream the point $\bmq \in \bbR^2$, given by
    \[ 
    \bmq = \left(
    3\cos\left(\pi + 2\pi i \sqrt{\varepsilon}\right),\,
    3\sin\left(\pi + 2\pi i \sqrt{\varepsilon}\right)
    \right),
    \]
    such that $\bmq$ is the point opposite $\bmpi$ through the origin such that $\norm{\bmpi - \bmq} = 4$.
    Without loss of generality, we assume that $i = n$. 
    Then, we have $\bmpi = (1, 0)$ and $\bmq = (-3, 0)$.
    If $b_i = 1$, then both $\bmpi$ and $\bmq$ are part of the input stream to Alice and Bob's instance of MEB. 
    In this case, it is easy to see that the MEB is that which has a diameter from $\bmpi$ to $\bmq$, meaning that the MEB has radius $2$.
    
    However, if $b_i = 0$, then $\bmpi$ is not part of the input to the MEB instance. Thus, it remains to show that $\norm{\bmpj - \bmq} \leq 4 - \Omega(\eps)$ for all other $j$ where $b_j = 1$. It is clear that we only need to consider $j$ such that $\bmpj$ is to the right of the origin, as any point $\bmpj$ left of the origin clearly satisfies this inequality. We can bound the square of this distance using the Pythagorean theorem as
    \begin{align*}
        \| \bmpj - \bmq\|^2 &= \paren{\cos\paren{2\pi j \sqrt{\varepsilon}} + 3}^2 + \paren{\sin\paren{2\pi j \sqrt{\varepsilon}}}^2\\
        &= 10 + 6\cos\paren{2\pi j \sqrt{\varepsilon}}.
    \end{align*}
    Thus, as this distance only depends on the cosine, and thus only on the $x$ coordinate of $\bmpj$, we can see that the distance from $\bmq$ is larger the further right a point is. Thus, the radius of the MEB is bounded above by the distance from $\bmq$ to $\bmp^{(1)}$ (or symmetrically $\bmp^{(n-1)}$). So, plugging in for $j$ we have that
    \[
        \| \bmp^{(1)} - \bmq\|^2 = 10 + 6\cos\paren{2\pi \sqrt{\varepsilon}}.
    \]
    It is a known result that $\cos(\theta) < 1 - \frac{4\theta^2}{\pi^2}$ for all $\theta \in (0, \frac{\pi}{2})$~\cite{BP22}. Thus we have that
    \[
        \| \bmp^{(1)} - \bmq\|^2 = 16 - 96 \eps
    \]
    which yields that the distance between $\bmp^{(1)}$ and $\bmq$ is $2 - \Omega(\varepsilon)$, giving our desired result. 
    
    We now extend this idea to an arbitrary dimension $d$ satisfying $d < \left(\sfrac{1}{\varepsilon}\right)^{0.999}$. 
    
    Without loss of generality, we assume $d$ is even. If it is not, we can simply disregard one of the dimensions. 
    Let $\varepsilon = \left(\frac{1}{n}\right)^{\sfrac{4}{d}}$. 
    As in the 2 dimensional case, we define Alice's $n$ points $\bmp^{(1)}, \bmp^{(2)}, \hdots, \bmp^{(n)} \in \bbR^d$ where each point $\bmpj$ is given by a sequence $\{j_k\}_{k=1}^{\sfrac{d}{2}}$ in $\{1, 2, \dots, \frac{1}{\sqrt{\varepsilon}}\}$, so that the $\ell$-th coordinate of $\bmpj$ is given by
    \[ \bmpj_\ell = \begin{cases}
        \sqrt{\frac{2}{d}}\cos\left(2 \pi j_{\frac{\ell+1}{2}} \sqrt{\varepsilon}\right) & \text{ if } \ell \text{ is odd,} \\
        \sqrt{\frac{2}{d}}\sin\left(2 \pi j_{\frac{\ell}{2}} \sqrt{\varepsilon}\right) & \text{ if } \ell \text{ is even.} \\
    \end{cases}  \]
    Informally, we pair up the axes of $d$-dimensional space, and choose the points on the unit circle which when projected onto the plane of a pair of axes will resemble the 2-dimensional case.
    Alice and Bob choose their points as in the 2-dimensional case. 
    Without loss of generality, we assume that $i = n$ and that $\bmpi$ is given by the sequence $i_k = \sqrt{\varepsilon}$ for all $k \in [\frac{d}{2}]$. Then, we have that $\bmpi$ is the point where the $\ell$-th coordinate is given by
    \[ \bmpi_\ell = \begin{cases}
    \sqrt{\frac{2}{d}} & \text{ if } \ell \text{ is odd,} \\
    0 & \text{ if } \ell \text{ is even.}
    \end{cases} \]
    We then have that $\bmq$ is point such that the $\ell$-th coordinate is given by
    \[ \bmq_\ell = \begin{cases}
    -3\sqrt{\frac{2}{d}} & \text{ if } \ell \text{ is odd,} \\
    0 & \text{ if } \ell \text{ is even.}
    \end{cases} \]
    As in the 2-dimensional case, it can easily be seen that if $b_i = 1$, then the MEB has radius 2. If $b_i = 0$, then the radius of the MEB is bounded above by the distance from $\bmq$ to $\bmpj$, where $\bmpj$ is given by $j_1 = 1$ and $j_k = \frac{1}{\sqrt{\varepsilon}}$ for all $k \ne 1$. We can now bound the squared distance.
    \begin{align*}
        \|\bmpj - \bmq\|^2 =\ & \sum_{k=1}^{d/2} 
        \left[ \sqrt{\frac{2}{d}}\cos\left( 2\pi j_k \sqrt{\varepsilon}\right) + 3\sqrt{\frac{2}{d}}\right]^2 + \left[\sqrt{\frac{2}{d}}\sin\left(2 \pi j_k \sqrt{\varepsilon}\right)\right]^2 \\
        =\ & \sum_{k=1}^{d/2} \frac{20}{d} + \frac{12}{d}\cos\left(2 \pi j_k \sqrt{\varepsilon}\right) \\
        =\ & \left(\frac{20}{d} + \frac{12}{d}\cos\left(2\pi \sqrt{\varepsilon}\right)\right) + \sum_{k=2}^{d/2} \left(\frac{20}{d} + \frac{12}{d} \cos\left(2 \pi j_k \sqrt{\varepsilon}\right)\right) \\
        \leq\ & \left(\frac{20}{d} + \frac{12}{d}\cos\left(2\pi \sqrt{\varepsilon}\right)\right) + \sum_{k=2}^{d/2} \left(\frac{20}{d} + \frac{12}{d}\right) && \text{as $\cos(\cdot) \leq 1$}\\
        =\ & \left(\frac{20}{d} + \frac{12}{d}\cos\left(2\pi \sqrt{\varepsilon}\right)\right) + \left(16 - \frac{32}{d}\right) \\
        \le\ & \left(\frac{20}{d} + \frac{12}{d} - \frac{196\varepsilon}{d}\right) + \left(16 - \frac{32}{d}\right) && \text{by~\cite{BP22}}\\
        =\ & 16 - \frac{192 \varepsilon}{d}. 
    \end{align*}
    
    This yields that the distance between $\bmpj$ and $\bmq$ is $2 - \Omega(\sfrac{\varepsilon}{d})$, giving a lower bound of $\left(\sfrac{1}{(\varepsilon d)}\right)^{\Omega(d)}$. In the regime where $d < \left(\sfrac{1}{\varepsilon}\right)^{0.999}$, this is $\left(\sfrac{1}{\varepsilon}\right)^{\Omega(d)}$, which is our desired result. 
\end{proof}

\subsection{Lower Bounds for the Linear SVM Problem}\label{svmlowersubsection}
In this section, we provide the following lower bounds for the space complexity of the $1$-pass streaming linear SVM problem via the communication complexity of the $1$-round linear SVM problem.

\ccsvmepsrestatable*

\begin{proof}
    To prove this, we prove a lower bound on any $1$-round communication protocol which yields a $(1 + \eps)$-approximation for the linear SVM problem.
    
    The proof follows the structure of the proof for Theorem~\ref{ccmebthm}, with a reduction from $\mathsf{Ind}_n$ to linear SVM. Alice transforms her bitstring into at most $n$ points labeled $-1$, and Bob transforms his index into a point labeled $+1$. We show then that if $b_i = 1$, then the separating hyperplane $(\bmu, b)$ of $P \cup \cbrak{\bmq}$ obtained by Alice and Bob has $\norm{\bmu} = 4$, while if $b_i = 0$ then $\norm{\bmu} \leq 4 - \Omega\paren{\sfrac{\eps}{d}}$.
    This therefore gives a $\paren{\sfrac{1}{\eps}}^{\Omega(d)}$ lower bound on the communication complexity of any $1$-round communication protocol for a $(1 + \eps)$-approximation, which gives the same lower bound on the space complexity of any $1$-pass streaming algorithm for the $(1 + \eps)$-approximation linear SVM problem. The full proof follows hereafter.

    Again, without loss of generality, we assume $d$ is even. If it is not, we can simply disregard one of the dimensions. 
        Let $\varepsilon = \left(\frac{1}{n}\right)^{\sfrac{4}{d}}$. We define Alice's $n$ points, which are all labelled $-1$
        \[
            \paren{\bmp^{(1)}, -1}, \paren{\bmp^{(2)}, -1}, \hdots, \paren{\bmp^{(n)}, -1} \in \bbR^d \times \cbrak{-1, +1}
        \]
        where each point $\bmpj$ is given by a sequence $\{j_k\}_{k=1}^{\sfrac{d}{2}}$ in $\{1, 2, \dots, \frac{1}{\sqrt{\varepsilon}}\}$, so that the $\ell$-th coordinate of $\bmpj$ is given by
        \[ \bmpj_\ell = \begin{cases}
            \sqrt{\frac{2}{d}}\cos\left(2 \pi j_{\frac{\ell+1}{2}} \sqrt{\varepsilon}\right) & \text{ if } \ell \text{ is odd,} \\
            \sqrt{\frac{2}{d}}\sin\left(2 \pi j_{\frac{\ell}{2}} \sqrt{\varepsilon}\right) & \text{ if } \ell \text{ is even.} \\
        \end{cases}  \]
    
    Alice and bob construct an input to the Linear SVM problem as follows: Alice adds to the input the labeled point $(\bmpj, -1)$ for each $j \in [n]$ such that $b_j = 1$, and Bob adds to the input the labeled point $(2 \bmpi, +1)$. 
    Call this instance of Linear SVM $\mcI$, and let $\paren{\bmu^*, b^*}$ be the optimal solution to this instance. 
    We make use of the following two lemmas, which will be proven later. 
    
    \begin{lemma} \label{epsLemmaOne}
    If $b_i = 1$, then $\norm{\bmu^*} = 4$.
    \end{lemma}
    
    \begin{lemma} \label{epsLemmaTwo}
    If $b_i = 0$, then $\norm{\bmu^*} \le 4 - \Omega(\sfrac{\varepsilon}{d})$. 
    \end{lemma}
    
    It then follows from the above lemmas that Alice and Bob can solve the indexing problem by obtaining a $(1 + \Theta(\sfrac{\varepsilon}{d}))$-approximation for $\mcI$. 
    Indexing has a lower bound of $\Omega(n)$ bits of communication, so we have that in general, a $(1 + \varepsilon)$-approximation of linear SVM must have a lower bound of $\left(\sfrac{1}{(\varepsilon d)}\right)^{\Omega(d)}$ bits of communication. In the regime where $d < \paren{\sfrac{1}{\varepsilon}}^{0.999}$, this is a lower bound of $\paren{\sfrac{1}{\varepsilon}}^{\Omega(d)}$.
    \end{proof}
    
    \begin{proof}[Proof of Lemma \ref{epsLemmaOne}] 
    Suppose that $b_i = 1$. 
    Let $\bmv$ be the vector given by
    \[ \bmv = \begin{bmatrix} 1 & 0 & 1 & 0 & \cdots &1 & 0 \end{bmatrix}^T. \]
    Without loss of generality, assume that $\bmpi$ is the point characterized by
    \[
        j_1 = j_2 = \cdots = j_{d/2} = 0,
    \]
    that is, $\bmpi = \sqrt{\frac{1}{2d}}\bmv$. 
    This can always be achieved via a rotation of the working space. 
    Consider the hyperplane given by $\sqrt{\frac{32}{d}}\bmv^T\bmx - 3 = 0$. 
    For any point $\bmpj$ which was potentially inserted into $\mcI$ with label $-1$ by Alice, we have
    \begin{align*}
        \sqrt{\frac{32}{d}}\bmv^T\bmpj - b &= \sqrt{\frac{16}{d^2}} \sum_{\ell=1}^{d/2} \cos(2\pi j_\ell \sqrt{\varepsilon}) - 3\\
        &\leq \frac{d}{2}\sqrt{\frac{16}{d^2}} - 3\\
        &= 1.
    \end{align*}
    Similarly, for the point $2 \bmpi$ inserted into $\mcI$ with label $+1$ by Bob, we have
    \begin{align*}
        2\sqrt{\frac{32}{d}}\bmv^T\bmpi - 3 &= d\sqrt{\frac{16}{d^2}} - 3\\
        &= 1.
    \end{align*}
    Thus, $(\sqrt{\frac{32}{d}}\bmv, 3)$ is a feasible solution to $\mcI$.
    To see that this is optimal, observe
    \[
        \|2\bmpi - \bmpi\| = \|\bmpi\|
        = \frac{1}{2}
        = \frac{2}{\norm{\sqrt{\frac{32}{d}}\bmv}}.
    \]
    Thus, we must have that $\|\bmu^*\|_2 = 4$.
    \end{proof} 
    
    \begin{proof}[Proof of Lemma \ref{epsLemmaTwo}]
    Suppose that $b_i = 0$. Let $\bmv$ be as defined in the proof of Lemma \ref{epsLemmaOne}. 
    For each $j \in [n]$, let $\bmqj$ be the orthogonal projection of $\bmpj$ onto the line spanned by $\bmv$. 
    Let $\mcJ$ be the instance of linear SVM consisting of the pairs $(\bmqj, -1)$ for all $j \in [n]$ such that $b_j = 1$ and the pair $(2\bmpi, +1)$. 
    It follows from orthogonality that if the hyperplane $\alpha \bmv^T \bmx - \beta = 0$ is a feasible solution to $\mcJ$, then it must also be a feasible solution to $\mcI$. 
    Using this fact, we proceed by showing there exists a hyperplane $\alpha \bmv^T \bmx - \beta = 0$ which is a feasible solution to $\mcJ$ and is such that $\norm{\alpha \bmv}_2 \le 4 - \Omega(\frac{\varepsilon}{d})$.
    First, we observe that for all $j \in [n]$
    \begin{align*}
        \bmqj &= \text{proj}_{\bmv}(\bmpj)\\
        &= \frac{\ip{\bmpj}{\bmv}}{\ip{\bmv}{\bmv}}\bmv\\
        &= \paren{\sum_{\ell=1}^{d/2} \frac{2}{d}\sqrt{\frac{1}{2d}}\cos(2\pi j_\ell \sqrt{\varepsilon})}\bmv.
    \end{align*}
    With this, we obtain that for each $j$ such that $b_j = 1$, we have
    \begin{align*}
        \norm{2\bmpi - \bmqj} =\ & \norm{\bmv}\abs{\sqrt{\frac{2}{d}} - \sum_{\ell=1}^{d/2} \frac{2}{d}\sqrt{\frac{1}{2d}}\cos(2\pi j_\ell \sqrt{\varepsilon})}  \\
        =\ & \sqrt{\frac{d}{2}}\abs{\sqrt{\frac{2}{d}} - \sum_{\ell=1}^{d/2} \frac{2}{d}\sqrt{\frac{1}{2d}}\cos(2\pi j_\ell \sqrt{\varepsilon})} \\
        =\ & \abs{1 - \sum_{\ell}^{d/2} \frac{1}{d}\cos(2 \pi j_\ell \sqrt{\varepsilon})} \\
        \ge\ & \abs{1 - \paren{\frac{1}{d}\paren{\frac{d}{2}-1} + \frac{1}{d}\cos(2\pi\sqrt{\varepsilon})}} \\
        =\ & \abs{\frac{1}{2} + \frac{1}{d}\paren{1 - \cos(2 \pi \sqrt{\varepsilon})}} \\
        \ge \ & \abs{\frac{1}{2} + \frac{1}{d}\paren{1 - \paren{1 - 16\varepsilon}}} \\
        =\ & \frac{1}{2} + \frac{16\varepsilon}{d}. 
    \end{align*} 
    Since all points in $\mcJ$ are in the line spanned by $\bmv$, and the points are separable, there must exist a feasible hyperplane $\alpha \bmv^T\bmx - \beta = 0$ such that
    \[ \frac{2}{\norm{\alpha \bmv}} = \min\left\{\norm{2\bmpi - \bmqj} : j \in [n] \text{ and } b_j = 1\right\} \ge \frac{1}{2} + \frac{16\varepsilon}{d}. \]
    Fnally, this gives that
    \[ \norm{\alpha \bmv} \le \frac{4d}{d + 32 \varepsilon} \le 4 - \frac{8\varepsilon}{d} = 4 - \Omega\paren{\frac{\varepsilon}{d}}. \]
    giving us the desired result.
\end{proof}

\ccsvmgammarestatable*
\begin{proof}
    To prove this, we prove a lower bound on any $1$-round communication protocol which which determines if a set of binary labeled points is $\gamma$-separable.
    
    The proof is similar to that of Theorem~\ref{ccsvmepsthm}, with a reduction from $\mathsf{Ind}_n$ to determining $\gamma$-separability. Alice transforms her bitstring exctly as in the proof of Theorem~\ref{ccsvmepsthm}, and Bob chooses a set of $d$ points, labeled $+1$, which lie in the unique hyperplane that is orthogonal to the line spanned by the point $\bmpi$, which is the point Alice would transform $b_i$ in the case where $b_i = 1$, and also contains $\bmpi$, such that $\bmpi$ lies in the segment between some pair of points chosen by Bob.

    Since $\bmpi$ can be written as a convex combination of two of Bob's points, it is clear that the points chosen by Alice and Bob are inseparable if $b_i = 1$, as then Alice will include $\bmpi$ in $P$.

    If however $b_i = 0$, then there will be a separating hyperplane parallel to Bob's set of points, and midway between that and the closest point of Alice. This hyperplane induces a margin of size $\Omega\paren{\sfrac{16\gamma}{d}}$, which gives a lower bound on the communication complexity of determining $\gamma$-separability of $\paren{\sfrac{1}{\gamma}}^{\Omega(d)}$. 
\end{proof}

\section{Proofs of Property \ref{P3} for Various Applications}

\subsection{MEB Problem}\label{meblemmasappendix}

\containsrestatable*

\begin{proof}
    Say point $\bmp_i$ was snapped to $\norm{\bmp_i} \bme_k$ for some point $\bme_k$ on the metric \epsnet{} and rounded up to $\bmq_j$. We know $\norm{\bmp_i} \bme_k$ is between $(1 + \eps)^l \bme_k$ and $(1 + \eps)^{l+1} \bme_k$ for some $l$, and so \linebreak${\norm{\bmp_i} \in \brak{(1 + \eps)^l, (1 + \eps)^{l+1}}}$.

    We first bound $d(\bmc, \bmp_i)$.
    \begin{align*}
        d(\bmc, \bmp_i) &\leq d(\bmc, \bmq_j) + d(\bmp_i, \bmq_j) && \text{by the triangle inequality}\\
        &\leq r + d(\bmp_i, \bmq_j) && \text{as $(\bmc,r)$ is the MEB of $Q$}\\
        &\leq r + d(\bmp_i, \norm{\bmp_i} \bme_k) + d(\norm{\bmp_i} \bme_k, \bmq_j) && \text{by the triangle inequality}\\
        &\leq r + \eps \norm{\bmp_i} + d(\norm{\bmp_i} \bme_k, \bmq_j) && \text{by the metric \epsnet{}}\\
        &\leq r + \eps \norm{\bmp_i} + (1 + \eps)^{l + 1} - (1 + \eps)^l\\
        &= r + \eps \norm{\bmp_i} + \eps (1 + \eps)^l\\
        &\leq r + \eps \norm{\bmp_i} + \eps \norm{\bmp_i}\\
        &= r + 2 \eps \norm{\bmp_i}.
    \end{align*}
    Now, we know $\norm{\bmq_j}$ is $(1 + \eps)^{l + 1}$ as it was snapped down. Since $\norm{\bmp_i} \in \brak{(1 + \eps)^l, (1 + \eps)^{l+1}}$, we can say that $\norm{\bmp_i} \leq \norm{\bmq_j}$. Plugging this in, we get
    \[
        d(\bmc, \bmp_i) \leq r + 2 \eps \norm{\bmq_j}.
    \]
    Now, since $\norm{\bmq_j}$ is $d(\bmq_j, \bmp_1)$ and $\bmp_1$ was unaffected by the snapping and rounding as the center of the \epsnet{}, both $\bmq_j$ and $\bmp_1$ are in $(\bmc, r)$. This means their distance is bounded by $2r$, meaning $\norm{\bmq_j} \leq 2r$. Plugging this in, we get
    \begin{align*}
        d(\bmc, \bmp_i) &\leq r + 4 \eps r\\
        &\leq (1 + 4 \eps) r. \qedhere
    \end{align*}
\end{proof}

\nottoobigrestatable*
\begin{proof}
    First, we know from the proof of Lemma~\ref{containslemma} that $d(\bmp_i, \bmq_j) \leq 2 \eps \norm{\bmp_i}$. As $\norm{\bmp_i}$ is $d(\bmp_i, \bmp_1)$, we can use the fact the fact that the MEB $(\bmc^*, r^*)$ contains all points $\bmp_i$, to bound this distance by $2 r^*$. Thus, $\norm{\bmp_i} \leq 2 r^*$. So we can say that  $d(\bmp_i, \bmq_j) \leq 4 \eps r^*$.

    Thus, since $(\bmc^*, r^*)$ encloses all the original points $\bmp_i \in P$, the ball $(\bmc^*, r^* + 4 \eps r^*)$ encloses all the snapped and rounded points $\bmq_j \in Q$.
    
    Thus, $r \leq r^* + 4 \eps r^*$ as $(\bmc, r)$ is the MEB for $Q$, and so $r \leq (1 + 4 \eps) r^*$.
\end{proof}

\subsection{Linear SVM Problem}\label{svmlemmasappendix}

\SVMcontainsrestatable*

\begin{proof}
    First, for any $i$, from our problem description, we know that $y_i \paren{\bmu^T \bmz_i - b} \geq 1$. Working with that, we get
    \begin{align*}
        y_i \paren{\bmu^T \bmz_i - b} &\geq 1\\
        \implies y_i \paren{\bmu^T \paren{\bmx_i + \bme_i} - b} &\geq 1\\
        \implies y_i \paren{\bmu^T \bmx_i - b} &\geq 1 - y_i \bmu^T \bme_i.
    \end{align*}
    Now, we look at $y_i \bmu^T \bme_i$. Since $y_i \in \cbrak{-1, +1}$, we can say that $y_i \bmu^T \bme_i \leq \pipe{\bmu^T \bme_i}$, which itself is bounded by $\norm{\bmu} \norm{\bme_i}$ by the Cauchy-Schwarz inequality.
    
    By construction, we know that $\norm{\bme_i} \leq \frac{\eps \gamma}{2}$.
    
    Moreover, since the original points are $\gamma$-separable, and snapping to the metric net moves each point at most $\frac{\eps \gamma}{2}$, we know that the snapped points are $(\gamma - \eps \gamma)$-separable, and so $\frac{\gamma}{2}$-separable as $\eps \leq \frac{1}{2}$. Thus, using the definitions of linear SVM
    \begin{align*}
        \frac{2}{\norm{\bmu}} &\geq \frac{\gamma}{2}\\
        \implies \norm{\bmu} &\leq \frac{4}{\gamma}.
    \end{align*}
    So, plugging these two results in we have that $y_i \bmu^T \bme_i \leq 2\eps$. Thus, we have that
    \begin{align*}
        y_i \paren{\bmu^T \bmx_i - b} &\geq 1 - 2\eps\\
        \implies y_i \paren{\frac{1}{1 - 2\eps} \paren{\bmu^T \bmx_i - b}} &\geq 1\\
        \implies y_i \paren{\paren{1 + 2\eps} \paren{\bmu^T \bmx_i - b}} &\geq 1 && \text{for small $\eps$}
    \end{align*}
    giving us the desired result.
\end{proof}

\SVMnottoobigrestatable*

\begin{proof}
    For small $\eps$, $\paren{1 + O(\eps)} \norm{\bmu^*}^2 \approx \norm{\paren{1 + O(\eps)} \bmu^*}^2$, so it suffices to show that
    \[
        \norm{\bmu}^2 \leq \norm{\paren{1 + O(\eps)} \bmu^*}^2.
    \]
    We know that for any $i$, $y_i \paren{\bmu^{*T} \bmx_i - b^*} \geq 1$. So, by the same argument as Lemma~\ref{SVMcontainslemma}
    \[
        y_i \paren{\paren{1 + 2\eps} \bmu^{*T} \bmz_i - \paren{1 + 2\eps}b} \geq 1.
    \]
    
    Thus, as $\paren{\paren{1 + 2\eps} \bmu^*, \paren{1 + 2\eps} b^*}$ is a solution to Problem~\ref{snappedSVM}, and $(\bmu,b)$ is the minimal solution to the optimization problem, $\norm{\bmu}^2 \leq \norm{\paren{1 + O(\eps)} \bmu^*}^2$ as desired.
\end{proof}

\subsection{Bounded LP Problems}\label{lplemmasappendix}

\LPcontainsrestatable*

\begin{proof}
    First, for any $i$, from our problem description, we know that $\bma_i'^T \bmx \leq b_i'$. Working with that, we get
    \begin{align*}
        \bma_i'^T \bmx &\leq b_i'\\
        \implies \paren{\bma_i + \bme_i}^T \bmx &\leq b_i + f_i\\
        \implies \bma_i^T \bmx &\leq b_i - \bme_i^T \bmx + f_i.
    \end{align*}
    Now, we look at $-\bme_i^T \bmx$. It is bounded by $\abs{\bme_i^T \bmx}$, which itself is bounded by $\norm{\bme_i} \norm{\bmx}$ by the Cauchy-Schwarz inequality. By construction, we know that $\norm{\bme_i} \leq O(\eps)$, and as given by our class of problem, $\norm{\bmx} \in O(1)$. So, plugging these two results in we have that $-\bme_i^T \bmx \leq O(\eps)$. Further, we have that $f_i$ is bounded by $O(\eps)$ by construction. Thus, together we have that
    \[
        \bma_i^T \bmx \leq b_i + O(\eps)
    \]
    as desired.
\end{proof}

\LPnottoobigrestatable*

\begin{proof}
    We know that for any $i$, $\bma_i^T \bmx^* \leq b_i$. So, by the same argument as Lemma~\ref{LPcontainslemma}
    \[
        \bma_i'^T \bmx^* \leq b_i' + O(\eps).
    \]
    Thus, $\bmx^*$ is a solution to the additive $\eps$-approximation of LP~\eqref{snappedLP}. So, if we denote the optimal solution to it as $\bmx'$, we have that $\bmc^T\bmx' \geq \bmc^T\bmx^*$.
    It then only remains to show that the optimal solution to the additive $\eps$-approximation of LP~\eqref{snappedLP} $\bmx'$ is not too much larger than the optimal solution to the exact LP~\eqref{snappedLP}. This follows because both $\norm{\bmc}$ and $\norm{\bmx}$ are bounded in $O(1)$. Allowing the constraints to be approximately satisfied by $O(\eps)$ can make the optimal point increase in norm by at most a $(1 + O(\eps))$ multiplicative factor, and thus an additive $O(\eps)$ factor for bounded norm. Thus, we have that $\bmc^T\bmx' \leq \bmc^T \bmx + O(\eps)$. Plugging this in, we get
    \[
        \bmc^T\bmx \geq \bmc^T \bmx^* - O(\eps)
    \]
    as desired.
\end{proof}

\subsection{Bounded SDP Problems}\label{sdplemmasappendix}

\SDPcontainsrestatable*

\begin{proof}
    First, for any $i$, from our problem description, we know that $\Frobinner{A'^{(i}}{X} \leq b'^{(i)}$. Working with that, we get
    \begin{align*}
        \Frobinner{A'^{(i}}{X} &\leq b'^{(i)}\\
        \implies \sum_{j,k} A'^{(i)}_{jk} X_{jk} &\leq b'^{(i)}\\
        \implies \sum_{j,k} \paren{\Ai_{jk} + \ei_{jk}} X_{jk} &\leq \bi + f^{(i)}\\
        \implies \sum_{j,k} \Ai_{jk} X_{jk} &\leq \bi + f^{(i)} + \sum_{j,k} -\ei_{jk} X_{jk}.
    \end{align*}
    First, we look at $f^{(i)}$. It is bounded by $\abs{f^{(i)}}$, which by construction is bounded by $\eps$. Secondly, we look at $\sum_{j,k} -\ei_{jk} X_{jk}$. That is bounded by $\sum_{j,k} \abs{\ei_{jk} X_{jk}} = \sum_{j,k} \abs{\ei_{jk}} \abs{X_{jk}}$. We can bound this quantity in two ways. First, by construction, $\abs{\ei_{jk}} \leq \frac{\eps}{\min(d,S)}$. Thus, we get that
    \begin{align*}
        \sum_{j,k} \abs{\ei_{jk}} \abs{X_{jk}} &\leq \sum_{j,k} \frac{\eps}{\minds} \abs{X_{jk}}\\
        &= \frac{\eps}{\minds} \sum_{j,k} \abs{X_{jk}}\\
        &= \frac{\eps}{\minds} \norm{X_{jk}} && \text{where $\norm{\cdot}$ is the Entry-wise $1$-norm}\\
        &\leq \frac{\eps}{\minds} d \norm{X_{jk}}_{*} && \text{where $\norm{\cdot}_{*}$ is the Schatten $1$-norm}\\
        &= \frac{d \eps}{\minds}.
    \end{align*}
    Secondly, for all $j,k$, we have that $\abs{X_{jk}} \leq \norm{X}_2 \leq \norm{X}_F \leq \norm{X}_* = 1$. We can thus plug this in to get
    \begin{align*}
        \sum_{j,k} \abs{\ei_{jk}} \abs{X_{jk}} &\leq \sum_{j,k} \abs{\ei_{jk}}\\
        &\leq S \frac{\eps}{\minds} && \text{since at most $S$ $\ei_{jk}$'s are nonzero.}
    \end{align*}
    Thus, putting these together, we have that $\sum_{j,k} \abs{\ei_{jk}} \abs{X_{jk}} \leq \eps$.
    So, plugging these in, we get the result
    \[
        \Frobinner{\Ai}{X} \leq \bi + 2\eps.
    \]
    Now, we can explore $\Frobinner{\Ai}{X + \frac{3 \eps}{d} I}$. We can see that
    \begin{align*}
        \Frobinner{\Ai}{X + \frac{3 \eps}{d} I} &= \sum_j \Ai_{jj} \paren{X_jj + \frac{3 \eps}{d}} + \sum_{j \neq k} \Ai_{jk} X_{jk}\\
        &= \Frobinner{\Ai}{X} + \frac{3 \eps}{d} \sum_j \Ai_{jj}\\
        &\leq \Frobinner{\Ai}{X} + \frac{3 \eps}{d} d \norm{\Ai}_2\\
        &\leq \bi + 2\eps + 3\eps\\
        &= \bi + 5\eps
    \end{align*}
    as desired.
\end{proof}

\SDPpsdrestatable*

\begin{proof}
    First, from our lattice construction, we know that $\bmz^T X \bmz \geq 0$ for $\bmz = \bmy + \bme$, where $\norm{\bme} \leq \frac{\eps}{d}$. If $\bme = \Vec{0}$ then we directly have the desired bounds. Assuming $\bme \neq \Vec{0}$, we get
    \begin{align*}
        \bmz^T X \bmz \geq 0\\
        \implies (\bmy + \bme)^T X (\bmy + \bme) \geq 0\\
        \implies \bmy^T X \bmy \geq - \bme^T X \bmy - \bmy^T X \bme - \bme^T X \bme.
    \end{align*}
    First, we look at $\bme^T X \bme$. Notice that this is bounded by $\abs{\bme^T X \bme}$, which is bounded by $\norm{\bme} \norm{X \bme}$ by Cauchy-Schwarz. Further, $\norm{X \bme}$ is bounded by $\norm{X}_2 \norm{\bme}$. We already saw in Lemma~\ref{SDPcontainslemma} that $\norm{X}_2$ is bounded by 1, and by construction $\norm{\bme} \leq \frac{\eps}{d}$. Thus $\bme^T X \bme \leq \paren{\frac{\eps}{d}}^2 \leq \frac{\eps}{d}$.

    Looking at $\bme^T X \bmy$ and $\bmy^T X \bme$, we can use similar logic to bound both by $\norm{\bme^T} \norm{\bmy} \norm{X}$. Since $\norm{\bmy} = 1$ by construction, these terms are then bounded by $\frac{\eps}{d}$.

    Plugging these in, we get the result
    \[
        \bmy^T X \bmy \geq -3\frac{\eps}{d}.
    \]
    Now, we can look at $\bmy^T (X + \frac{3 \eps}{d} I) \bmy$. We see that
    \begin{align*}
        \bmy^T (X + \frac{3 \eps}{d} I) \bmy &= \bmy^T X \bmy + \frac{3 \eps}{d} \bmy^T I \bmy\\
        &= \bmy^T X \bmy + \frac{3 \eps}{d} \norm{\bmy}^2\\
        &\geq -\frac{3 \eps}{d} + \frac{3 \eps}{d}\\
        &= 0
    \end{align*}
    as desired.
\end{proof}

\SDPnottoobigrestatable*

\begin{proof}
    We know that for any $i$, $\Frobinner{\Ai}{X^*} \leq \bi$. Further, $X^* \succeq 0$ So, by the same argument as Lemmas~\ref{SDPcontainslemma} and~\ref{SDPpsdlemma}
    \[
        \Frobinner{A'^{(i)}}{X^*} \leq b'^{(i)} + 2\eps \quad \text{and} \quad \bmz^T X^* \bmz \geq -\frac{3 \eps}{d}.
    \]
    Thus, $X^*$ is a solution to the additive $\eps$-approximation of the snapped LP. So, if we denote the optimal solution to it as $X'$, we have that $\Frobinner{C}{X'} \geq \Frobinner{C}{X^*}$.
    It then only remains to show that the optimal solution to the additive $\eps$-approximation of the snapped LP $X'$ is not too much larger than the optimal solution to the exact snapped LP. This follows because both $\norm{C}_F$ and $\norm{X}_F$ are bounded by $1$. the LP formulation of the SDP has bounded vectors $c$ and $x$. So, we can use the same analysis as done for the proof of~\ref{LPnottoobiglemma}. So, we have that $\Frobinner{C}{X'} \leq \Frobinner{C}{X} + O(\eps)$. Plugging this in, we get
    \[
        \Frobinner{C}{X} \geq \Frobinner{C}{X^*} - O(\eps).
    \]
    Finally, it only remains to show that
    \[
        \Frobinner{C}{X + \frac{3 \eps}{d}I} \geq \Frobinner{C}{X} - O(\eps),
    \]
    This also follows similarly to the analysis of Lemma~\ref{SDPcontainslemma}, such that
    \begin{align*}
         \Frobinner{C}{X + \frac{3 \eps}{d}I} &= \Frobinner{C}{X} + \frac{3 \eps}{d}\sum_j C_{jj}\\
         &\geq \Frobinner{C}{X} - \frac{3 \eps}{d} \abs{\sum_j C_{jj}}\\
         &\geq \Frobinner{C}{X} - \frac{3 \eps}{d} d \norm{C}_2\\
         &\geq \Frobinner{C}{X} - 3 \eps \norm{C}_F\\
         &\geq \Frobinner{C}{X} - 3 \eps
    \end{align*}
    as desired.
\end{proof}

\end{document}